\documentclass[conference, a4paper, 10pt]{IEEEtran}

\usepackage{times}
\usepackage[latin1]{inputenc}
\usepackage[T1]{fontenc}
\usepackage[english]{babel}
\usepackage{graphicx}
\usepackage{amsmath, amsfonts, dsfont,amssymb, graphicx, array, tabularx, booktabs}
\usepackage{amsthm}
\usepackage{epsf,epsfig}
\usepackage{setspace}
\usepackage{subfigure}
\usepackage[linesnumbered,ruled,vlined]{algorithm2e}

\SetCommentSty{mycommfont}

\usepackage{multirow}
\usepackage{url}
\usepackage{color}
\usepackage[nolist,printonlyused]{acronym}      
\usepackage{comment}
\usepackage{cite}
\usepackage{bm}
\usepackage{tikz}
\usetikzlibrary{decorations.pathreplacing}

\SetKwInOut{Parameter}{Parameter}

\usepackage{hyperref}
\usepackage{mathtools}

\usepackage{blindtext}

\usepackage{stfloats}

\newtheorem{thm}{Theorem} 

\newtheorem{cor}{Corollary}
\newtheorem{lem}{Lemma}
\newtheorem{prop}{Proposition}

\newcommand{\mx}[1]{\mathbf{#1}}
\newcommand{\bs}[1]{\boldsymbol{#1}}
\providecommand{\keywords}[1]{\textbf{\textit{Index terms---}} #1}

\definecolor{amber}{rgb}{1.0, 0.49, 0.0}
\definecolor{ao}{rgb}{0.0, 0.5, 0.0}

\def\R2#1{\textcolor{black}{#1}}
\def\R3#1{\textcolor{black}{#1}}

\begin{document}

\title{Sampling}
\singlespacing

\title{Optimizing Pilot Spacing in MU-MIMO Systems Operating Over Aging Channels}

\singlespacing
\author{
Sebastian Fodor$^\flat$, G\'{a}bor Fodor$^{\star\dag}$, Do\u{g}a G\"{u}rg\"{u}no\u{g}lu$^{\dag}$,
Mikl\'{o}s Telek$^{\ddag\sharp}$ \\
\small $^\flat$Stockholm University, Stockholm, Sweden. E-mail: \texttt{sebbifodor@fastmail.com}\\
\small $^\star$Ericsson Research, Stockholm, Sweden. E-mail: \texttt{Gabor.Fodor@ericsson.com}\\
\small $^\dag$KTH Royal Institute of Technology, Stockholm, Sweden. E-mail: \texttt{gaborf|dogag@kth.se}\\
\small $^\ddag$Budapest University of Technology and Economics, Budapest, Hungary. E-mail: \texttt{telek@hit.bme.hu}\\
\small $^\sharp$MTA-BME Information Systems Research Group, Budapest, Hungary. E-mail: \texttt{telek@hit.bme.hu}
}
\maketitle
\pagestyle{plain}


\begin{acronym}
  \acro{2G}{Second Generation}
  \acro{3G}{3$^\text{rd}$~Generation}
  \acro{3GPP}{3$^\text{rd}$~Generation Partnership Project}
  \acro{4G}{4$^\text{th}$~Generation}
  \acro{5G}{5$^\text{th}$~Generation}
  \acro{AA}{Antenna Array}
  \acro{AC}{Admission Control}
  \acro{AD}{Attack-Decay}
  \acro{ADSL}{Asymmetric Digital Subscriber Line}
	\acro{AHW}{Alternate Hop-and-Wait}
  \acro{AMC}{Adaptive Modulation and Coding}
	\acro{AP}{Access Point}
  \acro{APA}{Adaptive Power Allocation}
  \acro{AR}{autoregressive}
  \acro{ARMA}{Autoregressive Moving Average}
  \acro{ATES}{Adaptive Throughput-based Efficiency-Satisfaction Trade-Off}
  \acro{AWGN}{additive white Gaussian noise}
  \acro{BB}{Branch and Bound}
  \acro{BD}{Block Diagonalization}
  \acro{BER}{bit error rate}
  \acro{BF}{Best Fit}
  \acro{BLER}{BLock Error Rate}
  \acro{BPC}{Binary power control}
  \acro{BPSK}{binary phase-shift keying}
  \acro{BPA}{Best \ac{PDPR} Algorithm}
  \acro{BRA}{Balanced Random Allocation}
  \acro{BS}{base station}
  \acro{CAP}{Combinatorial Allocation Problem}
  \acro{CAPEX}{Capital Expenditure}
  \acro{CBF}{Coordinated Beamforming}
  \acro{CBR}{Constant Bit Rate}
  \acro{CBS}{Class Based Scheduling}
  \acro{CC}{Congestion Control}
  \acro{CDF}{Cumulative Distribution Function}
  \acro{CDMA}{Code-Division Multiple Access}
  \acro{CL}{Closed Loop}
  \acro{CLPC}{Closed Loop Power Control}
  \acro{CNR}{Channel-to-Noise Ratio}
  \acro{CPA}{Cellular Protection Algorithm}
  \acro{CPICH}{Common Pilot Channel}
  \acro{CoMP}{Coordinated Multi-Point}
  \acro{CQI}{Channel Quality Indicator}
  \acro{CRM}{Constrained Rate Maximization}
	\acro{CRN}{Cognitive Radio Network}
  \acro{CS}{Coordinated Scheduling}
  \acro{CSI}{channel state information}
  \acro{CSIR}{channel state information at the receiver}
  \acro{CSIT}{channel state information at the transmitter}
  \acro{CUE}{cellular user equipment}
  \acro{D2D}{device-to-device}
  \acro{DCA}{Dynamic Channel Allocation}
  \acro{DE}{Differential Evolution}
  \acro{DFT}{Discrete Fourier Transform}
  \acro{DIST}{Distance}
  \acro{DL}{downlink}
  \acro{DMA}{Double Moving Average}
	\acro{DMRS}{Demodulation Reference Signal}
  \acro{D2DM}{D2D Mode}
  \acro{DMS}{D2D Mode Selection}
  \acro{DPC}{Dirty Paper Coding}
  \acro{DRA}{Dynamic Resource Assignment}
  \acro{DSA}{Dynamic Spectrum Access}
  \acro{DSM}{Delay-based Satisfaction Maximization}
  \acro{ECC}{Electronic Communications Committee}
  \acro{EFLC}{Error Feedback Based Load Control}
  \acro{EI}{Efficiency Indicator}
  \acro{eNB}{Evolved Node B}
  \acro{EPA}{Equal Power Allocation}
  \acro{EPC}{Evolved Packet Core}
  \acro{EPS}{Evolved Packet System}
  \acro{E-UTRAN}{Evolved Universal Terrestrial Radio Access Network}
  \acro{ES}{Exhaustive Search}
  \acro{FDD}{frequency division duplexing}
  \acro{FDM}{Frequency Division Multiplexing}
  \acro{FER}{Frame Erasure Rate}
  \acro{FF}{Fast Fading}
  \acro{FSB}{Fixed Switched Beamforming}
  \acro{FST}{Fixed SNR Target}
  \acro{FTP}{File Transfer Protocol}
  \acro{GA}{Genetic Algorithm}
  \acro{GBR}{Guaranteed Bit Rate}
  \acro{GLR}{Gain to Leakage Ratio}
  \acro{GOS}{Generated Orthogonal Sequence}
  \acro{GPL}{GNU General Public License}
  \acro{GRP}{Grouping}
  \acro{HARQ}{Hybrid Automatic Repeat Request}
  \acro{HMS}{Harmonic Mode Selection}
  \acro{HOL}{Head Of Line}
  \acro{HSDPA}{High-Speed Downlink Packet Access}
  \acro{HSPA}{High Speed Packet Access}
  \acro{HTTP}{HyperText Transfer Protocol}
  \acro{ICMP}{Internet Control Message Protocol}
  \acro{ICI}{Intercell Interference}
  \acro{ID}{Identification}
  \acro{IETF}{Internet Engineering Task Force}
  \acro{ILP}{Integer Linear Program}
  \acro{JRAPAP}{Joint RB Assignment and Power Allocation Problem}
  \acro{UID}{Unique Identification}
  \acro{IID}{Independent and Identically Distributed}
  \acro{IIR}{Infinite Impulse Response}
  \acro{ILP}{Integer Linear Problem}
  \acro{IMT}{International Mobile Telecommunications}
  \acro{INV}{Inverted Norm-based Grouping}
	\acro{IoT}{Internet of Things}
  \acro{IP}{Internet Protocol}
  \acro{IPv6}{Internet Protocol Version 6}
  \acro{ISD}{Inter-Site Distance}
  \acro{ISI}{Inter Symbol Interference}
  \acro{ITU}{International Telecommunication Union}
  \acro{JOAS}{Joint Opportunistic Assignment and Scheduling}
  \acro{JOS}{Joint Opportunistic Scheduling}
  \acro{JP}{Joint Processing}
	\acro{JS}{Jump-Stay}
    \acro{KF}{Kalman filter}
  \acro{KKT}{Karush-Kuhn-Tucker}
  \acro{L3}{Layer-3}
  \acro{LAC}{Link Admission Control}
  \acro{LA}{Link Adaptation}
  \acro{LC}{Load Control}
  \acro{LOS}{Line of Sight}
  \acro{LP}{Linear Programming}
  \acro{LS}{least squares}
  \acro{LTE}{Long Term Evolution}
  \acro{LTE-A}{LTE-Advanced}
  \acro{LTE-Advanced}{Long Term Evolution Advanced}
  \acro{M2M}{Machine-to-Machine}
  \acro{MAC}{Medium Access Control}
  \acro{MANET}{Mobile Ad hoc Network}
  \acro{MC}{Modular Clock}
  \acro{MCS}{Modulation and Coding Scheme}
  \acro{MDB}{Measured Delay Based}
  \acro{MDI}{Minimum D2D Interference}
  \acro{MF}{Matched Filter}
  \acro{MG}{Maximum Gain}
  \acro{MH}{Multi-Hop}
  \acro{MIMO}{multiple input multiple output}
  \acro{MINLP}{Mixed Integer Nonlinear Programming}
  \acro{MIP}{Mixed Integer Programming}
  \acro{MISO}{Multiple Input Single Output}
  \acro{ML}{maximum likelihood}
  \acro{MLWDF}{Modified Largest Weighted Delay First}
  \acro{MME}{Mobility Management Entity}
  \acro{MMSE}{minimum mean squared error}
  \acro{MOS}{Mean Opinion Score}
  \acro{MPF}{Multicarrier Proportional Fair}
  \acro{MRA}{Maximum Rate Allocation}
  \acro{MR}{Maximum Rate}
  \acro{MRC}{maximum ratio combining}
  \acro{MRT}{Maximum Ratio Transmission}
  \acro{MRUS}{Maximum Rate with User Satisfaction}
  \acro{MS}{mobile station}
  \acro{MSE}{mean squared error}
  \acro{MSI}{Multi-Stream Interference}
  \acro{MTC}{Machine-Type Communication}
  \acro{MTSI}{Multimedia Telephony Services over IMS}
  \acro{MTSM}{Modified Throughput-based Satisfaction Maximization}
  \acro{MU-MIMO}{multiuser multiple input multiple output}
  \acro{MU}{multi-user}
  \acro{NAS}{Non-Access Stratum}
  \acro{NB}{Node B}
  \acro{NE}{Nash equilibrium}
  \acro{NCL}{Neighbor Cell List}
  \acro{NLP}{Nonlinear Programming}
  \acro{NLOS}{Non-Line of Sight}
  \acro{NMSE}{Normalized Mean Square Error}
  \acro{NORM}{Normalized Projection-based Grouping}
  \acro{NP}{Non-Polynomial Time}
  \acro{NRT}{Non-Real Time}
  \acro{NSPS}{National Security and Public Safety Services}
  \acro{O2I}{Outdoor to Indoor}
  \acro{OFDMA}{orthogonal frequency division multiple access}
  \acro{OFDM}{orthogonal frequency division multiplexing}
  \acro{OFPC}{Open Loop with Fractional Path Loss Compensation}
	\acro{O2I}{Outdoor-to-Indoor}
  \acro{OL}{Open Loop}
  \acro{OLPC}{Open-Loop Power Control}
  \acro{OL-PC}{Open-Loop Power Control}
  \acro{OPEX}{Operational Expenditure}
  \acro{ORB}{Orthogonal Random Beamforming}
  \acro{JO-PF}{Joint Opportunistic Proportional Fair}
  \acro{OSI}{Open Systems Interconnection}
  \acro{PAIR}{D2D Pair Gain-based Grouping}
  \acro{PAPR}{Peak-to-Average Power Ratio}
  \acro{P2P}{Peer-to-Peer}
  \acro{PC}{Power Control}
  \acro{PCI}{Physical Cell ID}
  \acro{PDF}{Probability Density Function}
  \acro{PDPR}{pilot-to-data power ratio}
  \acro{PER}{Packet Error Rate}
  \acro{PF}{Proportional Fair}
  \acro{P-GW}{Packet Data Network Gateway}
  \acro{PL}{Pathloss}
  \acro{PPR}{pilot power ratio}
  \acro{PRB}{physical resource block}
  \acro{PROJ}{Projection-based Grouping}
  \acro{ProSe}{Proximity Services}
  \acro{PS}{Packet Scheduling}
  \acro{PSAM}{pilot symbol assisted modulation}
  \acro{PSO}{Particle Swarm Optimization}
  \acro{PZF}{Projected Zero-Forcing}
  \acro{QAM}{Quadrature Amplitude Modulation}
  \acro{QoS}{Quality of Service}
  \acro{QPSK}{Quadri-Phase Shift Keying}
  \acro{RAISES}{Reallocation-based Assignment for Improved Spectral Efficiency and Satisfaction}
  \acro{RAN}{Radio Access Network}
  \acro{RA}{Resource Allocation}
  \acro{RAT}{Radio Access Technology}
  \acro{RATE}{Rate-based}
  \acro{RB}{resource block}
  \acro{RBG}{Resource Block Group}
  \acro{REF}{Reference Grouping}
  \acro{RLC}{Radio Link Control}
  \acro{RM}{Rate Maximization}
  \acro{RNC}{Radio Network Controller}
  \acro{RND}{Random Grouping}
  \acro{RRA}{Radio Resource Allocation}
  \acro{RRM}{Radio Resource Management}
  \acro{RSCP}{Received Signal Code Power}
  \acro{RSRP}{Reference Signal Receive Power}
  \acro{RSRQ}{Reference Signal Receive Quality}
  \acro{RR}{Round Robin}
  \acro{RRC}{Radio Resource Control}
  \acro{RSSI}{Received Signal Strength Indicator}
  \acro{RT}{Real Time}
  \acro{RU}{Resource Unit}
  \acro{RUNE}{RUdimentary Network Emulator}
  \acro{RV}{Random Variable}
  \acro{SAC}{Session Admission Control}
  \acro{SCM}{Spatial Channel Model}
  \acro{SC-FDMA}{Single Carrier - Frequency Division Multiple Access}
  \acro{SD}{Soft Dropping}
  \acro{S-D}{Source-Destination}
  \acro{SDPC}{Soft Dropping Power Control}
  \acro{SDMA}{Space-Division Multiple Access}
  \acro{SE}{spectral efficiency}
  \acro{SER}{Symbol Error Rate}
  \acro{SES}{Simple Exponential Smoothing}
  \acro{S-GW}{Serving Gateway}
  \acro{SINR}{signal-to-interference-plus-noise ratio}
  \acro{SI}{Satisfaction Indicator}
  \acro{SIP}{Session Initiation Protocol}
  \acro{SISO}{single input single output}
  \acro{SIMO}{Single Input Multiple Output}
  \acro{SIR}{signal-to-interference ratio}
  \acro{SLNR}{Signal-to-Leakage-plus-Noise Ratio}
  \acro{SMA}{Simple Moving Average}
  \acro{SNR}{signal-to-noise ratio}
  \acro{SORA}{Satisfaction Oriented Resource Allocation}
  \acro{SORA-NRT}{Satisfaction-Oriented Resource Allocation for Non-Real Time Services}
  \acro{SORA-RT}{Satisfaction-Oriented Resource Allocation for Real Time Services}
  \acro{SPF}{Single-Carrier Proportional Fair}
  \acro{SRA}{Sequential Removal Algorithm}
  \acro{SRS}{Sounding Reference Signal}
  \acro{SU-MIMO}{single-user multiple input multiple output}
  \acro{SU}{Single-User}
  \acro{SVD}{Singular Value Decomposition}
  \acro{TCP}{Transmission Control Protocol}
  \acro{TDD}{time division duplexing}
  \acro{TDMA}{Time Division Multiple Access}
  \acro{TETRA}{Terrestrial Trunked Radio}
  \acro{TP}{Transmit Power}
  \acro{TPC}{Transmit Power Control}
  \acro{TTI}{Transmission Time Interval}
  \acro{TTR}{Time-To-Rendezvous}
  \acro{TSM}{Throughput-based Satisfaction Maximization}
  \acro{TU}{Typical Urban}
  \acro{UE}{User Equipment}
  \acro{UEPS}{Urgency and Efficiency-based Packet Scheduling}
  \acro{UL}{uplink}
  \acro{UMTS}{Universal Mobile Telecommunications System}
  \acro{URI}{Uniform Resource Identifier}
  \acro{URM}{Unconstrained Rate Maximization}
  \acro{UT}{user terminal}
  \acro{VR}{Virtual Resource}
  \acro{VoIP}{Voice over IP}
  \acro{WAN}{Wireless Access Network}
  \acro{WCDMA}{Wideband Code Division Multiple Access}
  \acro{WF}{Water-filling}
  \acro{WiMAX}{Worldwide Interoperability for Microwave Access}
  \acro{WINNER}{Wireless World Initiative New Radio}
  \acro{WLAN}{Wireless Local Area Network}
  \acro{WMPF}{Weighted Multicarrier Proportional Fair}
  \acro{WPF}{Weighted Proportional Fair}
  \acro{WSN}{Wireless Sensor Network}
  \acro{WWW}{World Wide Web}
  \acro{XIXO}{(Single or Multiple) Input (Single or Multiple) Output}
  \acro{ZF}{zero-forcing}
  \acro{ZMCSCG}{Zero Mean Circularly Symmetric Complex Gaussian}
\end{acronym}

\begin{abstract}
In the uplink of multiuser multiple input multiple output (MU-MIMO) systems operating over aging channels, pilot spacing is crucial for acquiring channel state information and achieving high signal-to-interference-plus-noise ratio (SINR). Somewhat surprisingly, very few works examine the impact of pilot spacing on the correlation structure of subsequent channel estimates and the resulting quality of channel state information considering channel aging. In this paper, we consider a fast-fading environment characterized by its exponentially decaying autocorrelation function, and model pilot spacing as a sampling problem to capture the inherent trade-off between the quality of channel state information and the number of symbols available for information carrying data symbols. We first establish a quasi-closed form for the achievable asymptotic deterministic equivalent SINR when the channel estimation algorithm utilizes multiple pilot signals. Next, we establish upper bounds on the achievable SINR and spectral efficiency, as a function of pilot spacing, which helps to find the optimum pilot spacing within a limited search space. Our key insight is that to maximize the achievable SINR and the spectral efficiency of MU-MIMO systems, proper pilot spacing must be applied to control the impact of the aging channel and to tune the trade-off between pilot and data symbols.
\end{abstract}

\section{Introduction}
In wireless communications, pilot symbol-assisted channel estimation and prediction are used to achieve
reliable coherent reception, and thereby to provide a variety of high quality services in a spectrum efficient
manner. In most practical systems, the transmitter and receiver nodes acquire and predict channel state information
by employing predefined pilot sequences during the training phase, after which information symbols can be
appropriately modulated and precoded at the transmitter and estimated at the receiver.
Since the elapsed time between pilot transmissions and the transmit power level of pilot symbols have a
large impact on the quality of channel estimation, a large number of papers investigated the optimal
spacing and power control of pilot signals in both single and multiple antenna systems
\cite{Yan:01, Zhang:07B, Abeida:10, Hijazi:10, GH:12,Truong:13, Kong:2015, Chiu:15, Kashyap:17, Kim:20, Yuan:20, Fodor:21}.

Specifically in the uplink of \ac{MU-MIMO} systems, several papers proposed pilot-based channel estimation and
receiver algorithms assuming that the complex vector channel undergoes block fading, meaning that the channel is
constant
between two subsequent channel estimation instances
\cite{Couillet:2012, Wen:2013, Hoydis:13, Mallik:18}.
In the block fading
model, the evolution of the channel is memoryless in the sense that each channel realization is drawn independently
of previous channel instances from some characteristic distribution. While the block fading model is useful for
obtaining analytical expressions for the achievable \ac{SINR} and capacity \cite{Hoydis:13, Hanlen:2012}, it fails
to capture the correlation between subsequent channel realizations and the aging of the channel between estimation
instances \cite{Truong:13, Kong:2015, Yuan:20, Fodor:21}.

Due to the importance of capturing the evolution of the wireless channel in time, several papers developed
time-varying channel models, as an alternative to block fading models, whose states are advantageously estimated
and predicted by means of suitably spaced pilot signals. In particular,
a large number of related works assume that the wireless channel can be represented as an \ac{AR} process whose
states are estimated and predicted using Kalman filters, which exploit the correlation between subsequent
channel realizations \cite{Abeida:10, Hijazi:10, Truong:13, Kim:20, Fodor:21}.
These papers assume that the coefficients of the related \ac{AR} process are known, and the
current and future states of the process (and thereby of the wireless channel) can be well estimated.
Other important related works concentrate on estimating the coefficients of \ac{AR} processes based on
suitable pilot-based observations and measurements \cite{Mahmoudi:11, Xia:15, Esfandiari:20}.
In our recent work \cite{Fodor:21}, it was shown that when an \ac{AR} process is a good model of the
wireless channel and the \ac{AR} coefficients are well estimated, not only the channel estimation can
exploit the memoryful property of the channel, but also a new \ac{MU-MIMO} receiver can be
designed, which minimizes the \ac{MSE} of the received data symbols by exploiting the correlation
between subsequent channel states.
It is important to realize that the above references build on discrete time \ac{AR} models, in which the state transition matrix is an input of the model and can be estimated by some suitable system identification technique, such as the one
proposed in \cite{Esfandiari:20}.
However, these papers do not ask the question of how often
the channel state of a continuous time channel should be observed by suitably spaced pilot signals to realize a certain state transition matrix in the \ac{AR} model of the channel.

Specifically, a key characteristic of a continuous time Rayleigh fading environment is that the autocorrelation
function of the associated stochastic process is a zeroth-order Bessel function, which must be properly modelled \cite{Zheng:03, Wang:07}.
This requirement is problematic when developing discrete-time \ac{AR} models, 
since it is well-known
that Rayleigh fading cannot be perfectly modelled with any finite order \ac{AR} process
(since the autocorrelation function of discrete time \ac{AR} processes does not follow a Bessel function),
although the statistics of \ac{AR} process can approximate those of Rayleigh fading \cite{McGuire:05,Zheng:05}.

Recognizing the importance of modeling fast fading, including Rayleigh fading, channels with proper autocorrelation function as a basis for
pilot spacing optimization, papers \cite{Savazzi:09, Savazzi:09B} use a continuous time process as a representation of
the wireless channel, and address the problem of pilot spacing as a sampling problem. 
According to this approach, pilot placement can be considered as a sampling problem of the fading variations,
and the quality of the channel estimate is determined by the density and accuracy of channel sampling \cite{Savazzi:09B}.
However, these papers consider \ac{SISO} systems, do not deal with the problem of pilot and data power control, and
are not applicable to \ac{MU-MIMO} systems employing a \ac{MMSE} receiver, which was proposed in, for example, \cite{Fodor:21}.
On the other hand, paper \cite{Truong:13} analyzes the impact of channel aging on the performance of MIMO systems, without investigating the interplay between pilot spacing and the resulting state transition matrix of the \ac{AR} model of the fast fading channel.
The most important related works, their assumptions and key performance metrics are listed and compared with those
of the current paper in Table \ref{tab:tab1}.

\begin{table*}[h!]
	\centering
	\caption{Overview of Related Literature}
	\vspace{2mm}
	\label{tab:tab1}
	\footnotesize
	\begin{tabular}{
			|p{0.08\textwidth}|
			>{\centering}p{0.12\textwidth}|
			>{\centering}p{0.13\textwidth}|
			>{\centering}p{0.1\textwidth}|
			>{\centering}p{0.12\textwidth}|
			>{\centering}p{0.12\textwidth}|
			p{0.14\textwidth}|}
		\hline
		\textbf{Reference} & \textbf{Block fading vs. Aging channel} & \textbf{Is \ac{AR} modeling used?}
		& \textbf{Channel est.} & \textbf{SISO or MIMO receiver} & \textbf{Key performance indicators} & \textbf{~~Comment}  \\
		\hline

		Truong et al., \cite{Truong:13} & channel aging between pilots & discrete time AR approximating a Bessel func. & MMSE based on known AR params 
		& max. ratio combiner (MRC) receiver (not AR-aware) & average SINR, achievable rate & both UL and DL are considered \\
		\hline
        Zhang et al., \cite{Zhang:07B} & channel aging between pilots & discrete time AR approximating a Bessel f. & adaptive est. of AR params & SISO joint channel and data est.
        & BER & 
        AR(2) parameter estimation and demodulation 
        \\
        \hline
		Savazzi and Spagnolini \cite{Savazzi:09} & channel aging between pilots & 
		AR 
		channel evolution over 
		estimation instances  & interpolation based on multiple observations & SISO
		& average SINR and BER & power control is out of scope \\
		\hline
		Mallik et al., \cite{Mallik:18} & block fading channel & not applicable & MMSE channel estimation & SIMO with MRC
		& average SINR, symbol error probability & pilot/data power control is out of scope \\
		\hline
		Akin and Gursoy \cite{Akin:07} & channel aging between pilots & discrete time first order AR (Gauss-Markov) process & MMSE channel  estimation
		& SISO & achievable rate and bit energy $E_b/N_0$ & optimal power distribution and training period for SISO are derived \\
		\hline
		Chiu and Wu \cite{Chiu:15} & channel aging between plots & discrete time AR model approximating a Rayleigh fading & Kalman filter assisted channel estimation 
		& receiver structure is out-of-scope & MSE of channel est., 
		data rate, capacity 
		& pilot/data power control is out of scope \\
		\hline
		Fodor et al., \cite{Fodor:21} & no aging between pilots; correlated
		pilot intervals
		& discrete time AR(1) model & Kalman assisted ch. est. & AR(1)-aware MIMO MMSE receiver
		& MSE of the received data symbols & optimum pilot power control for AR(1) channels \\
		\hline
		Present paper & channel aging between pilots & AR(1) to model channel aging between pilots & MMSE interpolation by multiple 
		observations
		& MU-MIMO with MMSE receiver & average (det. equivalent) SINR & both pilot spacing and pilot/data power control are considered \\
		\hline
	\end{tabular}
\end{table*}

In this paper, we are interested in determining the average \ac{SINR} in the uplink of \ac{MU-MIMO} systems
operating in fast fading as a function of pilot spacing, pilot/data power allocation,
number of antennas and spatially multiplexed users. Specifically, we ask the following two important questions,
which are not answered by previous works:

\begin{itemize}
	\item
	What is the average \ac{SINR} in a closed or quasi-closed form in the uplink of \ac{MU-MIMO} systems
	in fast fading in the presence of antenna correlation?
	How does the average \ac{SINR} depend on pilot spacing and pilot/data power control?
	\item
	What is the optimum pilot spacing and pilot/data power allocation as a function of
	the number of antennas and the Doppler frequency associated with the continuous time fast fading channel?
\end{itemize}

In the light of the above discussion and questions, the main contributions of the present paper
are as follows:

\begin{itemize}
	\item
	Proposition \ref{Prop:SINR} derives the asymptotic deterministic equivalent \ac{SINR} of any user
	in a \ac{MU-MIMO} system
	in every data slot for fast fading channels that can be characterized by an associated \ac{AR} process;
	\item
	Theorem \ref{thm:upperbound} and Proposition \ref{UpperB1} establish an upper bound on the achievable \ac{SINR} as a function of
	pilot spacing, which is instrumental for determining the optimum pilot spacing.
	\item
	Proposition \ref{UpperB2}, building on Proposition \ref{UpperB1}, provides and upper bound on the average achievable spectral efficiency, which is instrumental in limiting the search space for the optimal frame size as a function of the Doppler frequency.
\end{itemize}
In addition, we believe that the engineering insights drawn from the numerical studies are
useful when designing pilot spacing, for example in the form of determining the number of reference signals in an uplink frame structure, for \ac{MU-MIMO} systems.

Specifically, to answer the above questions, we proceed as follows. In the next section, we present our system model,
which admits correlated wireless channels between any of the single-antenna mobile terminal and the receive
antennas of the \ac{BS}.
Next, Section \ref{Sec:ChannelE} focuses on channel estimation, which is based on subsequent pilot-based
measurements and an \ac{MMSE}-interpolation for the channel states in between estimation instances.
Section \ref{Sec:SINR} proposes an algorithm to determine the average \ac{SINR}.
Section \ref{Sec:PilotSpacing}  studies the impact of pilot spacing and power control on the achievable \ac{SINR} and the \ac{SE} of all users in the system.
That section investigates the impact of pilot spacing 
on the achievable \ac{SINR}
and establishes an upper bound on this \ac{SINR}. We show that this upper bound is monotonically
decreasing as the function of pilot spacing. This property is very useful, because it enables to
limit the search space of the possible pilot spacings when looking for the optimum pilot spacing in Section \ref{Sec:Alg}.
That section also considers the special case when the channel coefficients associated with the different
receive antennas are uncorrelated and identically distributed. It turns out that in this special
case a simplified \ac{SINR} expression can be derived.
Section \ref{Sec:Num} presents numerical results and discusses engineering insights.
Finally, Section \ref{Sec:Conc} draws conclusions.

\section{System Model}

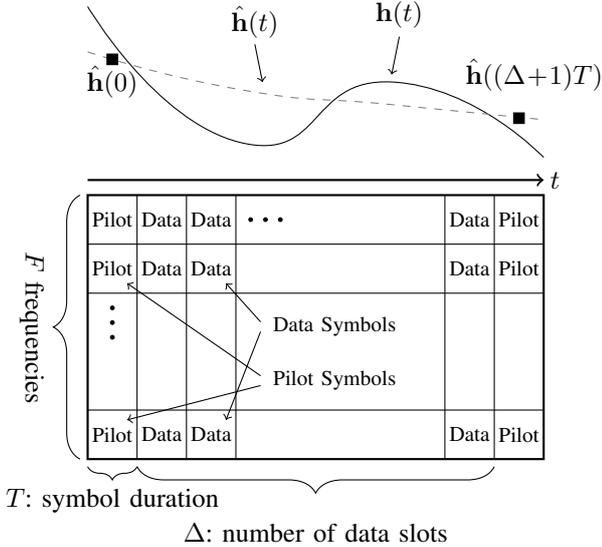
\begin{figure}
\centering
\begin{tikzpicture}[scale = 1]
  \def\x{0.65}
  \def\y{0.07}

  \draw[black, thick] (0,0) rectangle (6,3.5);

  \draw (\x,0) -- (\x,3.5);
  \draw (2*\x,0) -- (2*\x,3.5);
  \draw (3*\x,0) -- (3*\x,3.5);
  \draw (6-\x,0) -- (6-\x,3.5);
  \draw (6-2*\x,0) -- (6-2*\x,3.5);

  \draw (0, \x) -- (6, \x);
  \draw (0, 3.5-\x) -- (6, 3.5-\x);
  \draw (0, 3.5-2*\x) -- (6, 3.5-2*\x);

  \node at (\x/2, \x/2) {\footnotesize Pilot};
  \node at (\x/2, 3.5-\x/2) {\footnotesize Pilot};
  \node at (\x/2, 3.5-3*\x/2) {\footnotesize Pilot};

  \node at (3*\x/2, \x/2) {\footnotesize Data};
  \node at (3*\x/2, 3.5-\x/2) {\footnotesize Data};
  \node at (3*\x/2, 3.5-3*\x/2) {\footnotesize Data};

  \node at (5*\x/2, \x/2) {\footnotesize Data};
  \node at (5*\x/2, 3.5-\x/2) {\footnotesize Data};
  \node at (5*\x/2, 3.5-3*\x/2) {\footnotesize Data};

  \node at (6-3*\x/2, \x/2) {\footnotesize Data};
  \node at (6-3*\x/2, 3.5-\x/2) {\footnotesize Data};
  \node at (6-3*\x/2, 3.5-3*\x/2) {\footnotesize Data};

  \node at (6-\x/2, \x/2) {\footnotesize Pilot};
  \node at (6-\x/2, 3.5-\x/2) {\footnotesize Pilot};
  \node at (6-\x/2, 3.5-3*\x/2) {\footnotesize Pilot};

  \draw [->] (3.5*\x, 1.5*\x) -- (0.8*\x, 0.8*\x);
  \draw [->] (3.5*\x, 1.7*\x) -- (0.8*\x, 3.5-1.8*\x);
  \node at (5*\x, 1.6*\x) {\footnotesize Pilot Symbols};

  \draw [->] (3.5*\x, 2.6*\x) -- (2.8*\x, 0.8*\x);
  \draw [->] (3.5*\x, 2.8*\x) -- (2.8*\x, 3.5-1.8*\x);
  \node at (5*\x, 2.7*\x) {\footnotesize Data Symbols};

  \filldraw [black] (\x/2,3.5-2.3*\x) circle (0.7pt);
  \filldraw [black] (\x/2,3.5-2.6*\x) circle (0.7pt);
  \filldraw [black] (\x/2,3.5-2.9*\x) circle (0.7pt);

  \filldraw [black] (3.3*\x,3.5-\x/2) circle (0.7pt);
  \filldraw [black] (3.6*\x,3.5-\x/2) circle (0.7pt);
  \filldraw [black] (3.9*\x,3.5-\x/2) circle (0.7pt);

  \draw [decorate,decoration={brace,amplitude = 10pt}] (6-\x,-0.1) -- (\x,-0.1)
  node [black,midway,yshift=-25pt] {$\Delta$: number of data slots};
  \draw [decorate,decoration={brace,amplitude = 5pt}] (\x,-0.1) -- (0,-0.1)
  node [black,midway,yshift=-13pt] {$T$: symbol duration};

  \draw [decorate,decoration={brace,amplitude = 10pt}] (-0.1, 0) -- (-0.1,3.5)
  node [black,midway,xshift=-18pt] {\rotatebox{270}{$F$ frequencies}};

  \draw [->, thick] (0,3.7) -- (6,3.7);
  \node at (6.15, 3.7) {$t$};

  \draw plot [smooth, tension=1] coordinates { (0, 6) (2,4.2) (4,5) (6,4)};
  \draw [dashed, gray] plot [smooth, tension=1] coordinates { (0,5.4) (2,4.9) (4,4.7) (6,4.5)};

  \filldraw (\x/2-\y,5.3-\y) rectangle (\x/2+\y, 5.3+\y);
  \node at (\x/2, 5) {$\hat{\mx{h}}(0)$};
  \filldraw (6-\x/2-\y,4.52-\y) rectangle (6-\x/2+\y, 4.52+\y);
  \node at (6.2-\x/2, 5.1) {$\hat{\mx{h}}((\Delta\!+\!1)T)$};

  \draw [->] (4.1,5.6) -- (4,5.1);
  \node at (4.1, 5.9) {$\mx{h}(t)$};

  \draw [->] (2.2,5.5) -- (2.3,5);
  \node at (2.2, 5.8) {$\hat{\mx{h}}(t)$};
\end{tikzpicture}
\caption{Pilot (P) and data (D) symbols in the time-frequency domains of the system in the $(0,(\Delta+1)T)$ interval. The solid line above the
time-frequency resource grid represents the continuous time complex channel $\mx{h}(t)$, while the dashed line represents
the \ac{MMSE} channel estimate $\hat{\mx{h}}(t)$.
Notice that in each time slot of length $T$ all symbols are either pilot or
data symbols.}
\label{fig:Model}
\end{figure}

\begin{table}[t]
\caption{System Parameters}
\vspace{2mm}
\label{tab:notation}
\footnotesize
\begin{tabularx}{\columnwidth}{|X|X|}
\hline
\hline
\textbf{Notation} & \textbf{Meaning} \\
\hline
\hline
$K$ & Number of \ac{MU-MIMO} users \\
\hline
$N_r$ & Number of antennas at the BS \\
\hline
$F$ & Number of frequency channels used for pilot and data transmission within one slot \\
\hline
$\Delta$ & number of data slots in a data-pilot cycle \\
\hline
$\tau_p=F, \tau_d=\Delta F$ & Number of pilot/data symbols within a coherent set of subcarriers  \\
\hline
$\mx{s}\in \mathds{C}^{\tau_p \times 1}$ & Sequence of pilot symbols\\
\hline
$x$ & Data symbol \\
\hline
$P_p, P$ & Pilot power per symbol, data power per symbol\\
\hline
$\mx{Y}^p(t) \in \mathds{C}^{N_r \times \tau_p}, y(t) \in \mathds{C}^{N_r}$ & Received pilot and data signal at time $t$, respectively  \\
\hline
$\alpha$ & Large scale fading between the mobile station and the base station \\
\hline
$\mx{C} \in \mathds{C}^{N_r \times N_r}$ & Stationary covariance matrix of the fast fading channel \\
\hline
$\mx{h}(t), \hat{\mx{h}}(t) \in \mathds{C}^{N_r}$ & Fast fading channel and estimated channel \\
\hline
$\bs{\varepsilon}(t) \in \mathds{C}^{N_r}, \bs{\Sigma} \in \mathds{C}^{N_r \times N_r}$
& Channel estimation error and its covariance matrix\\
\hline
$\mx{G}^\star$ & Optimal MU-MIMO receiver. \\
\hline
$f_D$ & Maximum Doppler frequency \\
\hline
$T$ & Slot duration \\
\hline
\end{tabularx}
\end{table}

\subsection{Uplink Pilot Signal Model}
By extending the single antenna channel model of \cite{Savazzi:09},
each transmitting \ac{MS} uses a single time slot to send $F$ pilot symbols,
followed by $\Delta$ time slots, each of which containing $F$ data symbols according to Figure \ref{fig:Model}.
Each symbol is transmitted within a coherent time slot of duration $T$.
Thus, the total frame duration is $(1+\Delta) T$, such that each frame consists of 1 pilot
and $\Delta$ data time slots, which we will index with $i=1 \ldots \Delta$.
User-$k$ transmits each of the $F$ pilot symbols with transmit power $P_{p,k}$, and each data
symbol in slot-$i$ with transmit power $P_k(i),~k=1 \ldots K$.
To simplify notation, in the sequel we tag User-1, and will
drop index $k=1$ when referring to the tagged user.

Assuming that the coherence bandwidth accommodates at least $F$ pilot symbols,
this system allows to create $F$ orthogonal pilot sequences.
To facilitate spatial multiplexing
and \ac{CSIR} acquisition at the \ac{BS}, the \acp{MS} use orthogonal complex
sequences, such as shifted Zadoff-Chu sequences of length $\tau_p=F$, which we denote as:
\begin{align}
\mathbf{s} &\triangleq \left[s_1,...,s_{\tau_p}\right]^T \in \mathds{C}^{{\tau_p \times 1}},
\end{align}
whose elements satisfy 
$|s_i|^2 = 1$. 
Under this assumption, the system can spatially multiplex $K\leq F$ \acp{MS}.
Focusing on the received pilot signal from the tagged user at the \ac{BS},
the received pilot signal takes the form of \cite{Fodor:21}:
\begin{align}
\mathbf{Y}^p(t)
&=
\alpha \sqrt{P_{p}}\mathbf{h}(t) \mathbf{s}^T +\mathbf{N}(t) ~~ \in \mathds{C}^{N_r \times \tau_p},
\label{eqn:received_training_seq}
\end{align}
where 
$\mathbf{h}(t) \in \mathds{C}^{N_r \times 1} \sim \mathcal{CN}(\mathbf{0},\mathbf{C})$, that is,
$\mathbf{h}(t)$ is a 
complex normal distributed column vector
with mean vector $\mathbf{0}$ and covariance matrix $\mx{C} \in \mathds{C}^{N_r \times N_r}$. 
Furthermore, $\alpha$ denotes
large scale fading, $P_p$ denotes the pilot power of the tagged user,
and $\mathbf{N}(t)\in \mathds{C}^{N_r \times \tau_p}$
is the 
\ac{AWGN} with element-wise variance $\sigma_p^2$.
It will be convenient to introduce $\mathbf{\tilde Y}^p(t)$ by stacking the columns of  $\mathbf{Y}^p(t)$ as:
\begin{align}
\mathbf{\tilde Y}^p(t)=\textbf{vec}\big(\mathbf{Y}^p(t)\big)=\alpha\sqrt{P_p} \mathbf{S} \mathbf{h}(t) +\mathbf{\tilde N}(t) ~\in \mathds{C}^{\tau_p N_r \times 1},
\label{eqn:received_training_seq2}
\end{align}
where $\textbf{vec}$ is the column stacking vector operator,
$\mathbf{\tilde Y}^p(t)$, $\mathbf{\tilde N}(t) \in \mathds{C}^{\tau_p N_r \times 1}$
and
$\mathbf{S} \triangleq \mathbf{s}\otimes \mathbf{I}_{N_r} \in \mathds{C}^{\tau_p N_r \times N_r}$
is such that $\mathbf{S}^H\mathbf{S}=\tau_p\mathbf{I}_{N_r}$,
where $\mathbf{I}_{N_r}$ is the identity matrix of size $N_r$.
\vspace{-2mm}
\subsection{Channel Model}
\vspace{-1mm}
In \eqref{eqn:received_training_seq},
the channel $\mx{h}(t)$ evolves continuously according to a multivariate complex stochastic process
with stationary covariance matrix $\mx{C}$.
That is, for symbol duration $T$, the channel ($\mx{h}(t)$) evolves according to the following \ac{AR} process:
\begin{align}
\label{eq:AR1}
\mx{h}(t+T) &= \mx{A} \mx{h}(t) + \bs{\vartheta}(t),
\end{align}
where the transition matrix of the \ac{AR} process is denoted by $\mx{A}$.
This \ac{AR} model has been commonly used to approximate Rayleigh fading channels in e.g. \cite{Baddour05}.
Equation \eqref{eq:AR1} implies that the autocorrelation function of the channel process is:
\begin{align}
\mathds{E}\left(\mx{h}(t)\mx{h}^H(t+iT)\right) &= \mx{C}\left(\mx{A}^H\right)^i, \quad \forall i.
\end{align}
Consequently, the autocorrelation function of
the fast fading channel ($\mx{h}(t)$) is modelled as:
\begin{align}
\label{eq:autocorr}
\mx{R}(i) \triangleq
\mathds{E}\left(\mx{h}(t)\mx{h}^H(t+iT)\right) &= \mx{C} e^{\mx{Q}^H i T },
\end{align}
where matrix $\mx{Q}$ describes the correlation decay, such that:
\begin{align}
e^{\mx{Q}T} = \mx{A}.
\end{align}

Similarly, for user $k$,
\begin{align}
\label{eq:autocorrk}
\mx{R}_k(i) \triangleq
\mathds{E}\left(\mx{h}_k(t)\mx{h}_k^H(t+iT)\right) &= \mx{C}_k e^{\mx{Q}_k^H i T },
\end{align}

In each pilot slot,
the \ac{BS} utilizes \ac{MMSE} channel estimation to obtain the channel estimate of each user, as it will be detailed in Section \ref{Sec:ChannelE}.
Without loss of generality, to simplify the notation, hereafter we assume that the time unit is $T$ and $iT=i$.

\subsection{Data Signal Model}
When spatially multiplexing $K$ \ac{MU-MIMO} users,
the received data signal at the \ac{BS} at time $t$ is \cite{Fodor:21}:

\begin{align}
\mathbf{y}(t)
&=
\underbrace{\mathbf{\alpha} \mathbf{h}(t) \sqrt{P} x(t)}_{\text{tagged user}}
+ \underbrace{\sum_{k=2}^K \mathbf{\alpha}_{k} \mathbf{h}_k(t) \sqrt{P_{k}} x_{k}(t)}_{\text{co-scheduled MU-MIMO users}}
+\mathbf{n}_d(t),
\label{eq:mumimo2}
\end{align}
\noindent where $\mathbf{y}(t)\in \mathds{C}^{N_r \times 1}$;
and $x_k(t)$ denotes the transmitted data symbol of User-$k$
at time $t$ with transmit power $P_k$.
Furthermore, $\mathbf{n}_d(t)~\sim \mathcal{CN}\left(\mx{0},\sigma_d^2 \mathbf{I}_{N_r}\right)$
is the \ac{AWGN} at the receiver.
\vspace{-2mm}
\section{Channel Estimation}
\label{Sec:ChannelE}
In this section, we are interested in calculating the \ac{MMSE} estimation of the
channel in each slot $i$, based on received pilot signals, as a function of the frame
size corresponding to pilot spacing (see $\Delta$ in Figure \ref{fig:Model}).
Note that estimating the channel at the receiver can be based on
multiple received pilot signals both before and after the actual data slot $i$.
While using pilot signals that are received before data slot $i$ requires to
store the samples of the received pilot, using pilot signals that arrive after
data slot $i$ necessarily induces some delay in estimating the transmitted data symbol.
In the numerical section, we will refer to specific channel estimation strategies
as, for example, "1 before, 1 after" or "2 before, 1 after" depending on the number
of utilized pilot signals received prior to or following data slot $i$
for \ac{CSIR} acquisition.
In the sequel we use the specific case of "2 before, 1 after" to illustrate the
operation of the \ac{MMSE} channel estimation scheme,
that is when the receiver uses the pilot signals
$\mathbf{\tilde Y}^p(-\Delta-1)$, $\mathbf{\tilde Y}^p(0)$, and $\mathbf{\tilde Y}^p(\Delta+1)$
for \ac{CSIR} acquisition.
We are also interested in determining the distribution of the
resulting channel estimation error, whose covariance matrix, denoted by
$\mx{Z}(\Delta,i)$, will play an important role in subsequently determining
the deterministic equivalent of the \ac{SINR}.
\vspace{-2mm}
\subsection{\ac{MMSE} Channel Estimation and Channel Estimation Error}
As illustrated in Figure \ref{fig:Model},
in each data slot $i$, the \ac{BS} utilizes the \ac{MMSE} estimates of the channel
obtained in the neighboring pilot slots, for example at $(-\Delta-1)$, $0$ and $(\Delta+1)$, 
using the respective received pilot signals according to \eqref{eqn:received_training_seq2}, that is
$\mathbf{\tilde Y}^p\big((-\Delta-1)\big)$, $\mathbf{\tilde Y}^p(0)$ and $\mathbf{\tilde Y}^p\big((\Delta+1)\big)$, using the following lemma.

\begin{lem}
\label{lem:mmsechannel}	
The \textup{MMSE} channel estimator approximates the autoregressive fast fading channel in time slot $i$
based on the received pilots at $(-\Delta-1)$, $0$ and $(\Delta+1)$ as
\begin{align}
\label{eq:hmmse}
&\mathbf{\hat h}_{\textup{MMSE}}(\Delta,i)=
\mx{H^\star}(\Delta,i) \mathbf{\hat Y}^p(\Delta),
\end{align}
where
\begin{align*}
\mx{H^\star}(\Delta, i) = \frac{1}{\alpha \sqrt{P_p} \tau_p} \mx{E}(\Delta,i). \big(\mx{M}(\Delta)+\mx{\Sigma}_3\big)^{-1}.(\mx{s}^H\otimes \mx{I}_{3N_r}),
\end{align*}

$\mathbf{\hat Y}^p(\Delta) \triangleq \begin{bmatrix}
\mathbf{\tilde Y}^p((-\Delta-1)) \\
\mathbf{\tilde Y}^p(0) \\
\mathbf{\tilde Y}^p((\Delta+1))
\end{bmatrix}$~~and~~$\mx{\Sigma}_3\triangleq \frac{\sigma_p^2}{\alpha^2P_p \tau_p} \mx{I}_{3N_r}$,
\begin{align}
\label{eq:E}
\mx{E}(\Delta,i) &\triangleq
\begin{bmatrix}
\mx{R}(\Delta \!+\! 1 \!+\! i) & \mx{R}(i) &
\mx{R}(\Delta \!+\! 1 \!-\! i)
\end{bmatrix},
\\
\label{eq:M}
\mx{M}(\Delta) &\triangleq
\begin{bmatrix}
\mx{C} & \mx{R}(\Delta \!+\! 1) & \mx{R}(2\Delta \!+\! 2)  \\
\mx{R}^H(\Delta \!+\! 1) & \mx{C} & \mx{R}(\Delta \!+\! 1) \\
\mx{R}^H(2\Delta \!+\! 2) & \mx{R}^H(\Delta \!+\! 1) & \mx{C}\\
\end{bmatrix}.
\end{align}

\end{lem}

\begin{proof}
The \ac{MMSE} channel estimator aims at minimizing the \ac{MSE} between the channel estimate
$\mathbf{\hat h}_{\textrm{MMSE}}(\Delta, i) = \mathbf{H}^\star(\Delta, i) \mathbf{\hat Y}^p(\Delta)$ and the channel $\mathbf{h}(i)$, that is
\begin{align}
\mathbf{H^\star}(\Delta, i)= \text{arg} \min_{\mathbf{H}} \mathds{E}_{\mathbf{h},\mathbf{n}}\{ ||\mathbf{H} \mathbf{\hat Y}^p(\Delta) - \mathbf{h}(i)||^2 \}.
\end{align}
The solution of this quadratic optimization problem is
$\mathbf{H^\star}(\Delta, i)= \mx{a}(\Delta, i)^H \mx{F}^{-1}(\Delta)$
with
\begin{align}
\label{eq:Fandb}
\mx{F}(\Delta) &\triangleq \mathds{E}_{\mathbf{h},\mathbf{n}}\left(\mx{\hat Y}^p(\Delta) \left(\mx{\hat Y}^p(\Delta)\right)^H\right), \\
\mx{a}(\Delta, i) &\triangleq \mathds{E}_{\mathbf{h},\mathbf{n}}\left( \mx{\hat Y}^p(\Delta) \mx{h}^{H}(i)\right).
\end{align}
Let
\begin{align}
\mx{\bar h}(\Delta) \triangleq \begin{bmatrix}
\mx{h}((-\Delta-1)) \\
\mx{h}(0) \\
\mx{h}((\Delta+1))
\end{bmatrix} \nonumber
\end{align}
and
\begin{align}
\mx{\tilde{\bar{N}}}(\Delta) \triangleq \begin{bmatrix}
\mx{\tilde{N}}((-\Delta-1)) \\
\mx{\tilde{N}}(0) \\
\mx{\tilde{N}}((\Delta+1))
\end{bmatrix}.
\end{align}
Using $\mx{\bar h}(\Delta)$, we have
\begin{align}
\begin{bmatrix}
\mx{S}\mx{h}((-\Delta-1)) \\
\mx{S}\mx{h}(0) \\
\mx{S}\mx{h}((\Delta+1))
\end{bmatrix}=
\begin{bmatrix}
\mx{S}  & \mx{0} & \mx{0} \\
\mx{0} & \mx{S}  & \mx{0} \\
\mx{0} & \mx{0} & \mx{S}
\end{bmatrix}\mx{\bar h}(\Delta)
\nonumber \\
~~~~~~~~= (\mx{I}_3\otimes\mx{S}) \mx{\bar h}(\Delta)=  (\mx{s}\otimes \mx{I}_{3N_r}) \mx{\bar h}(\Delta).
\end{align}
Since $\mx{\bar h}(\Delta)$ and $\mx{\tilde{\bar{N}}}(\Delta)$ are independent and
\begin{align}
\mathds{E}_{\mathbf{h},\mathbf{n}}(\mx{\bar h}(\Delta)\mx{\bar h}^H(\Delta))&=\mx{M}(\Delta),\\
\mathds{E}_{\mathbf{h},\mathbf{n}}(\mx{h}(i)\mx{\bar h}^H(\Delta))&=\mx{E}(\Delta,i).
\end{align}
Therefore, for $\mx{F}(\Delta)$ and $\mx{a}(\Delta,i)$, we have
\begin{align*}
\mx{F}(\Delta) &=
\mathds{E}_{\mathbf{h},\mathbf{n}}
\left(\left(\alpha\sqrt{P_p} (\mx{I}_3\otimes\mx{S}) \mx{\bar{h}}(\Delta) +\mx{\tilde{\bar{N}}}(\Delta)\right)\right.\nonumber \\
&~~~~~~~~~~~~\cdot \left.\left(\alpha\sqrt{P_p} (\mx{I}_3\otimes\mx{S}) \mx{\bar{h}}(\Delta) +\mx{\tilde{\bar{N}}}(\Delta)\right)^H\right)\\
&=\alpha^2P_p(\mx{I}_3\otimes\mx{S})\mx{M}(\Delta)\left(\mx{I}_3\otimes\mx{S}^H\right)+ \sigma_p^2 \mathbf{I}_{3N_r\tau_p},\\
\mx{a}(\Delta, i) &= \mathds{E}_{\mathbf{h},\mathbf{n}}\left( \mx{\hat Y}(\Delta) \mx{h}^{H}(i)\right) \\
&=\mathds{E}_{\mathbf{h},\mathbf{n}}
\left(\left(\alpha\sqrt{P_p} (\mx{I}_3\otimes\mx{S}) \mx{\bar{h}}(\Delta) +\mx{\tilde{\bar{N}}}(\Delta)\right)
\mx{h}^{H}(i)\right)\\
&= \alpha\sqrt{P_p} ~(\mx{I}_3\otimes\mx{S})~\mx{E}(\Delta,i)^T,
\end{align*}
which yields Lemma \ref{lem:mmsechannel}.
\end{proof}

The \ac{MMSE} estimate of the channel is then expressed as:
\begin{align}
\nonumber
&\mathbf{\hat h}_{\textrm{MMSE}}(\Delta, i)=\mathbf{H^\star}(\Delta, i) \mathbf{\hat Y}^p(\Delta) ~~~~ \\
&~~~~=\mathbf{H^\star}(\Delta, i)
\left(\alpha\sqrt{P_p} (\mx{I}_3\otimes\mx{S}) \mx{\bar{h}}(\Delta) +\mx{\tilde{\bar{N}}}(\Delta)\right)
\nonumber
\\
&~~~~=
\frac{1}{\alpha\sqrt{P_p} \tau_p}
\mx{E}(\Delta, i) \left(   \mx{M}(\Delta) + \mx{\Sigma}_{3} \right)^{-1}
\nonumber
\\&~~~~~~~~~~.
\left(\alpha\sqrt{P_p} \tau_p \mx{\bar{h}}(\Delta) + \left(\mx{I}_3 \otimes \mx{S}^H\right) \mx{\tilde{\bar{N}}}(\Delta)\right).
\label{eq:MMSEt}
\end{align}

Next, we are interested in deriving the distribution of the estimated channel and
the channel estimation error, since these will be important for understanding the
impact of pilot spacing on the achievable \ac{SINR} and spectral efficiency of
the \ac{MU-MIMO} system. To this end, the following two corollaries of
Lemma \ref{lem:mmsechannel} and \eqref{eq:MMSEt} will be important in the sequel.
\begin{cor}
\label{cor:rmmse}	
The estimated channel $\mathbf{\hat h}_{\textup{MMSE}}(\Delta, i)$ is a circular symmetric complex normal distributed vector
$\mathbf{\hat h}_{\textup{MMSE}}(\Delta, i) \sim 
\mathcal{CN}\Big(\mathbf{0},\mathbf{\hat \Phi}_{\textup{MMSE}}(\Delta, i)\Big)$,
with
\begin{align}
\label{eq:rmmse}
\mathbf{\hat \Phi}_{\textup{MMSE}}(\Delta, i) \triangleq
& \mathds{E}_{\mathbf{h},\mathbf{n}} \{\mathbf{\hat h}_{\textup{MMSE}}(\Delta, i) \mathbf{\hat h}_{\textup{MMSE}}^H(\Delta, i)\}
 \nonumber \\
=& \mx{E}(\Delta, i) \big(\mx{M}(\Delta)+\mx{\Sigma}_3\big)^{-1} \mx{E}(\Delta, i)^H.
\end{align}

\end{cor}

\begin{proof}
Equation \eqref{eq:rmmse} follows directly from \eqref{eq:MMSEt}.	
\end{proof}

An immediate consequence of Corollary \ref{cor:rmmse}
is the following corollary regarding the covariance of the channel estimation
error, as a function of pilot spacing.
\begin{cor}
\label{Cor:ChEstError}
The channel estimation error in slot $i$, $\mx{\hat{h}}_{\textup{MMSE}}(\Delta, i)-\mx{h}(\Delta, i)$,
is complex normal distributed with zero mean vector and
covariance matrix given by:
\begin{align}
\label{eq:Z}
\mx{Z}(\Delta, i) \triangleq  \mx{C} -
\mx{E}(\Delta, i) \big(\mx{M}(\Delta)+\mx{\Sigma}_3\big)^{-1}
\mx{E}(\Delta, i)^H.
\end{align}
\end{cor}

In the following section we will calculate the \ac{SINR} of the received data symbols.
For simplicity of notation, we use $\mx{\hat h}_{\textup{MMSE}}(\Delta, i)=\mx{\hat h}(\Delta, i)$, and  introduce
$$\mx{b}(\Delta, i) \triangleq \alpha\sqrt{P(i)}\mx{\hat{h}}(\Delta, i)$$
with covariance matrix
\begin{align}
\label{eq:Phi}
\bs{\Phi}(\Delta, i) & \triangleq
\mathds{E}\left( \mx{b}(\Delta, i)\mx{b}^H(\Delta, i) \right)  \nonumber \\
&=\mathds{E}\left( \Big(\alpha\sqrt{P(i)} \mx{\hat{h}}(\Delta, i) \Big)\Big(\alpha\sqrt{P(i)} \mx{\hat{h}}(\Delta, i) \Big)^H  \right)  \nonumber \\
&=\alpha^2P(i)(\mx{C}-\mx{Z}(\Delta, i)).
\end{align}
\subsection{Summary}

This section derived the \ac{MMSE} channel estimator (Lemma \ref{lem:mmsechannel})
that uses the received pilot signals
both before and after a given data slot $i$ and depends on the frame size $\Delta$ (pilot spacing).
As important corollaries of the channel estimation scheme, we established the distribution
of both the estimated channel (Corollary \ref{cor:rmmse})
and the associated channel estimation error in each data
slot $i$ (Corollary \ref{Cor:ChEstError}), as functions of both the employed pilot spacing and pilot power. These
results serve as a starting point for deriving the achievable \ac{SINR} and spectral
efficiency.

\section{\ac{SINR} Calculation}
\label{Sec:SINR}

\subsection{Instantaneous \ac{SINR}}
We start with recalling an important lemma from \cite{Fodor:22}, which calculates the instantaneous \ac{SINR}
in an \ac{AR} fast fading environment
when the \ac{BS} uses the \ac{MMSE} estimation of the fading channel, and employs the optimal linear
receiver:
\begin{align}
\label{eq:Gstar2}
\mx{G}^\star(\Delta, i) &=
\mx{b}^H(\Delta,i) \mx{J}^{-1}(\Delta,i),
\end{align}
where $\mx{J}(\Delta, i) \in \mathds{C}^{N_r \times N_r}$
is defined as
\begin{align}
\nonumber
\mx{J}(\Delta, i)
&\triangleq
\sum_{k=1}^K  \mx{b}_k(\Delta, i) \mx{b}_k^H(\Delta, i) + \bs{\beta}(\Delta, i),
\end{align}
where
\begin{align}
\label{eq:beta}
\bs{\beta}(\Delta, i) \triangleq \sum_{k=1}^K \alpha_k^2 P_k  \mx{Z}_k(\Delta, i) + \sigma^2_d \mx{I}_{N_r}.
\end{align}

When using the above receiver, which minimizes the \ac{MSE} of the received data symbols in the
presence of channel estimation errors, the following result from \cite{Fodor:22} will be useful in the sequel:
\begin{lem}[See \cite{Fodor:22}, Lemma 3]
Assume that
the receiver employs \textup{\ac{MMSE}} symbol estimation,
that is it employs the optimal linear receiver $\mx{G}^\star(\Delta, i)$ given in \eqref{eq:Gstar2}. 
Then the instantaneous \textup{\ac{SINR}} of the estimated data symbols of the tagged user,
$\gamma(\Delta, i)$ 
is given as:
%
\begin{equation}
\label{eq:lemma2Eq}
\gamma(\Delta, i) 
=
\mx{b}^H(\Delta, i) \mathbf{\bar{J}}^{-1}(\Delta, i) \mx{b}(\Delta, i),
\end{equation}
where
\vspace{-1mm}
\begin{equation}
\mathbf{\bar{J}}(\Delta, i) \triangleq \mathbf{J}(\Delta, i)- \mx{b}(\Delta, i)\mx{b}^H(\Delta, i).
\end{equation}
\end{lem}
For the \ac{AR} fading case considered in this paper,
based on the definitions of $\mx{b}(\Delta, i)$, $\mx{J}(\Delta, i)$ and $\mx{\bar{J}}(\Delta, i)$,
the instantaneous \ac{SINR} of the tagged user is then expressed as:
\begin{align}
\gamma(i) &  =  \mx{b}^H(\Delta, i) \mathbf{\bar{J}}^{-1}(\Delta, i) \mx{b}(\Delta, i)
  \nonumber\\
 &=
\textup{tr}\left( \mx{b}(\Delta, i)\mx{b}^H(\Delta, i) \mathbf{\bar{J}}^{-1}(\Delta, i) \right).
\end{align}

\subsection{Slot-by-Slot Deterministic Equivalent of the \ac{SINR} as a Function of Pilot Spacing $\Delta$}
We can now prove the following important proposition that gives the asymptotic deterministic equivalent
of the instantaneous \ac{SINR} in data slot $i$, $\bar{\gamma}(\Delta, i)$, when the number of antennas $N_r$ approaches infinity.
This asymptotic equivalent \ac{SINR} gives a good approximation of averaging the instantaneous \ac{SINR} of the tagged user \cite{Wagner:2012, Hoydis:13, Fodor:22}.

\begin{prop}
\label{Prop:SINR}
The asymptotic deterministic equivalent \ac{SINR} of the tagged user in data slot $i$ can be calculated as:
\begin{align}
\label{eq:hoydis1}
\bar{\gamma}(\Delta, i) &= \textup{tr}\Big( \bs{\Phi}(\Delta, i)\mx{T}(\Delta, i)  \Big),
\end{align}
where $\mx{T}(\Delta, i)$ is defined as: 
\begin{align}
\label{eq:hoydis2}
\mx{T}(\Delta, i) \triangleq \left( \sum_{m = 2}^{K} \frac{\bs{\Phi}_m(\Delta, i)}{1 + \delta_m(\Delta, i)} + \bs{\beta}(\Delta, i) \right)^{-1},
\end{align}
and $\delta_m(\Delta, i)$ are the solutions of the following system of $K$ equations
\begin{align}
\label{eq:hoydis3}
\delta_m(\Delta, i) &= \textup{tr}\left(\bs{\Phi}_m(\Delta, i)
\left( \sum_{l = 2}^{K} \frac{\bs{\Phi}_l(\Delta, i)}{1 + \delta_l(\Delta, i)} + \bs{\beta}(\Delta, i)  \right)^{-1} \right)
\end{align}
for $\forall m = 1, \ldots, K$.
\end{prop}
The above system of $K$ equations gives the deterministic equivalent of the \ac{SINR} of the tagged user,
and a different set of $K$ equations must be used for each user.
\begin{proof}
The $\mx{b}_k(\Delta, i)$ vectors are independent for $k = 1 \ldots K$,
and the covariance matrix of $\mx{b}_k(\Delta, i)$ is $\bs{\Phi}_k(\Delta, i)$
(c.f. \eqref{eq:Phi}).
We can then express the expected value of the \ac{SINR} of the tagged user as follows:
\begin{align}
\label{eq:gammaB}
& \bar{\gamma}(\Delta, i) \triangleq \mathds{E}\Big(\gamma(\Delta, i)\Big)   \\
&= \mathds{E}\left( \textup{tr}\left( \bs{\Phi}(\Delta, i)  \left( \sum_{l = 2}^K  \mx{b}_l(\Delta, i) \mx{b}_l^H(\Delta, i)  +\bs{\beta}(\Delta, i) \right)^{\!\!-1}  \right) \right).
\nonumber
\end{align}
The proposition is established
by invoking Theorem \ref{thm:upperbound} in \cite{Wagner:2012}, which is applicable in multiuser systems and
gives the value of the deterministic equivalent of $\bar{\gamma}(\Delta, i)$
implicitly using a system of $K$ equations and noticing that
$\bar{\gamma}(\Delta, i) = \delta_1(\Delta, i)$, since $\delta_1(\Delta, i)=\textup{tr}\big( \bs{\Phi}(\Delta, i) \mx{T}(\Delta, i)\big)$ according to \eqref{eq:hoydis3}.
\end{proof}

\subsection{Summary}
This section established the instantaneous slot-by-slot \ac{SINR} of a tagged user ($\bar \gamma(i)$) of a \ac{MU-MIMO} system operating over 
a fast fading channels modelled as \ac{AR} processes,
by applying our previous result obtained for discrete-time \ac{AR} channels
reported in \cite{Fodor:21}.
Next, we invoked Theorem \ref{thm:upperbound} in \cite{Wagner:2012}, to establish the deterministic
equivalent \ac{SINR} for each slot, as a function of the frame size (pilot spacing) $\Delta$, see Proposition \ref{Prop:SINR}.
These results serve as a basis for formulating the pilot spacing optimization problem over the frame size and pilot power as optimization variables.

\section{Pilot Spacing and Power Control}
\label{Sec:PilotSpacing}
In this section, we study the impact of pilot spacing and power control on the achievable \ac{SINR}
and the \ac{SE} of all users in the system.
The asymptotic \ac{SE} of the $i$-th data symbol of user $k$ is
\begin{align}
\label{eq:SE}
\text{SE}_k(\Delta, i) \triangleq \log\Big(1 + \bar{\gamma}_k(\Delta, i)\Big),
\end{align}
where $\bar{\gamma}_k(\Delta, i)$ denotes the average \ac{SINR} of  user $k$ when sending the $i$-th data symbol, and when $\Delta$ data symbols are sent between every pair of pilot symbols.
Consequently, the 
average \ac{SE}
of  user $k$ over the $(\Delta+1)$ slot long frame is
\begin{align}
\frac{\sum_{i=1}^{\Delta}\text{SE}_k(\Delta, i)}{\Delta + 1},
\end{align}
which can be optimized over $\Delta$.
More importantly, 
the aggregate average \ac{SE}
of the \ac{MU-MIMO} system for the $K$ users can be expressed as:
\begin{align}
\text{SE}(\Delta) = \frac{\sum_{k=1}^K \sum_{i=1}^{\Delta}\text{SE}_k(\Delta,i)}{\Delta + 1}.
\label{eq:multiopt}
\end{align}

\subsection{An Upper Bound of the Deterministic Equivalent \ac{SINR} and the \ac{SE}}

Let us assume that $\mx{Q}_k=q_k \mx{I}_{N_r}$, that is the
channel vector $\mx{h}_k(t)$ consists of independent \ac{AR} processes in the spatial domain, implying that:
\begin{align}
\label{eq:r1}
\mx{R}_k(i) \triangleq
\mathds{E}\left(\mx{h}_k(t)\mx{h}_k^H(t+i)\right) &= \mx{C}_k e^{q_k^* i},
\end{align}
where $q_k$ is a scalar, $q_k^*$ denotes complex conjugation, and let $\bar{q}_k\triangleq Re(q_k)<0$.

Note that the exponential approximation of the autocorrelation function of the fast fading process expressed
in \eqref{eq:r1} is related to the Doppler frequency of Rayleigh fading through:
\begin{align}
\label{eq:RayleighApprox}
\underbrace{\mx{C}J_0(2 \pi f_D i)}_{\text{True autocorrelation of Rayleigh fading}} &\approx \mx{R}(i),
\end{align}
where $J_0(.)$ is the zeroth order Bessel function \cite{Wang:03}.
Based on the exponential approximation of this Rayleigh fading process in \eqref{eq:r1}, the Doppler frequency of the approximate model is obtained from $2 \pi f_D i = Re(q_k^* i)$, i.e. $f_D = 2 \pi/\bar{q}_k$.

To optimize \eqref{eq:multiopt}, we first find an upper bound of $\text{SE}_k(\Delta,i)$
via an upper bound of $\bar{\gamma}_k(\Delta, i)$.
To simplify the notation, the following discussion refers to the tagged user, and later we utilize that the same relations hold for all users.
We introduce the following upper bound of $\bar{\gamma}(\Delta, i)$:
\begin{align}
\label{eq:upper}
\bar{\gamma}^{(u)}(\Delta, i) \!&\triangleq \textup{tr}\!\left(\! \bs{\Phi}^{(u)}(\Delta, i)  \left(\sum_{l = 1}^K \!\alpha_l^2 P_l \mx{Z}^{(u)}_l(\Delta, i) \!+\! \sigma_d^2 \mx{I_{N_r}}\right)^{\!\!\!-1}  \!\right)\!,
\end{align}
where $\mx{Z}^{(u)}(\Delta, i)$ and $\bs{\Phi}^{(u)}(\Delta, i)$ are given by
\begin{align}
\label{eq:zu}
\mx{Z}^{(u)}(\Delta, i) &\triangleq \mx{C} - \rho(\Delta, i) \mx{C} \left(\eta  \mx{C} + \bs{\Sigma} \right)^{-1} \mx{C},  \\
\label{eq:phiu}
\bs{\Phi}^{(u)}(\Delta, i) &\triangleq \alpha^2 P \rho(\Delta, i) \mx{C} \left( \eta \mx{C} + \bs{\Sigma} \right)^{-1} \mx{C},
\end{align}
with $\eta$ being a constant, $\mx{\Sigma}\triangleq \frac{\sigma_p^2}{\alpha^2P_p \tau_p} \mx{I}_{N_r}$ and
\begin{align}
\label{eq:rho}
\rho(\Delta, i) &\triangleq e^{2\bar{q} (\Delta + 1 + i)} + e^{2\bar{q}  i} + e^{2\bar{q} (\Delta + 1 - i)}.
\end{align}

\begin{thm}
\label{thm:upperbound}
If $\bar{q}<0$ and
\begin{align}
\label{eq:etacond}
0<\eta<\frac{1}{2}\left(2+a^2-a\sqrt{8+a^2}\right),
\end{align}
with $a\triangleq e^{2\bar{q}}$
then  $\bar{\gamma}(\Delta, i)\leq \bar{\gamma}^{(u)}(\Delta, i)$.
\end{thm}
\begin{proof}
We prove the theorem based on the following inequalities
\begin{align}
\bar{\gamma}(\Delta, i)
&\overset{(a)}{\leq}
\textup{tr}\left( \bs{\Phi}(\Delta, i) \bs{\beta}(\Delta, i)^{-1} \right) \nonumber \\
&\overset{(b)}{\leq}
   \textup{tr}\left( \bs{\Phi}^{(u)}(\Delta, i) \bs{\beta}(\Delta, i)^{-1} \right)
\overset{(c)}{\leq}
   \bar{\gamma}^{(u)}(\Delta, i),  \label{eq:gammabound}
\end{align}
which are proved in consecutive lemmas.
\end{proof}

\begin{lem}
\label{lem:AB}
Let $\mx{A}$, $\mx{B}$ and $\mx{C}$ be positive definite matrices and $\mx{D}$ be any matrix, such that $\mx{A} \preceq \mx{B}$ (i.e. $\mx{B}-\mx{A}$ is a positive semidefinite matrix), then
\begin{align}
\label{eq:ABinv}
\mx{A}^{-1} &\succeq \mx{B}^{-1},
\\
\label{eq:AB3}
\textup{tr}\left( \mx{D}^H \mx{A} \mx{D} \right) &\leq \textup{tr}\left(\mx{D}^H \mx{B} \mx{D} \right) \\
\label{eq:AB1}
\textup{tr}\left( \mx{A} \mx{C} \right) &\leq \textup{tr}\left(\mx{B} \mx{C} \right)
\\
\label{eq:AB2}
\textup{tr}\left( \mx{C} \mx{A}^{-1} \right) &\geq \textup{tr}\left( \mx{C} \mx{B}^{-1} \right).
\end{align}
\end{lem}
\begin{proof}
$\mx{A}^{-1} \succeq \mx{B}^{-1}$ is given in \cite[p. 495, Corollary 7.7.4(a)]{Horn2013}.
\eqref{eq:AB3} follows from the fact that $\mx{D}^H (\mx{B}-\mx{A}) \mx{D}$ is a positive semidefinite matrix since $\mx{B}-\mx{A}$ is a positive semidefinite matrix and
for any $\mx{x}$
\begin{align}
\mx{x}^H \mx{D}^H (\mx{B}-\mx{A}) \mx{D} \mx{x} = \mx{y}^H (\mx{B}-\mx{A}) \mx{y} \geq 0
\end{align}
where $\mx{y} \triangleq \mx{D} \mx{x}$.
Let $\mx{C}=\mx{D}^H \mx{D}$ be the Cholesky decomposition of $\mx{C}$ then
\eqref{eq:AB1} and \eqref{eq:AB2} follows from \eqref{eq:AB3}, by utilizing the cyclic property of the trace operator.
\end{proof}

\begin{lem}
\label{lem:beta}
For $\bar{q}<0$ and $\eta$ satisfying \eqref{eq:etacond},
the following relation holds
\begin{align}
&\mx{E}(\Delta, i) \big(\mx{M}(\Delta,i)+\mx{\Sigma}_3\big)^{-1} \mx{E}(\Delta, i)^H
\nonumber \\&
~~~~~~~~~~~~~~\preceq \rho(\Delta, i) \mx{C} \left( \eta \mx{C} \!+\! \bs{\Sigma} \right)^{-1} \mx{C}
\label{eq:mrel}
\end{align}
\end{lem}
\begin{proof}
The proof is in Appendix \ref{Sec:L4}.
\end{proof}

Having prepared with Lemma \ref{lem:AB} and Lemma \ref{lem:beta}, we can prove the (a), (b) and (c) inequalities in \eqref{eq:gammabound} by Lemma \ref{lem:det} ((a) part) and Lemma \ref{lem:ineq} ((b) and (c) parts) as follows.
\begin{lem}
\label{lem:det}
The deterministic equivalent \ac{SINR} of the tagged user satisfies
\begin{align}
& \bar{\gamma}(\Delta, i) \leq   \textup{tr}\left( \bs{\Phi}(\Delta, i) \bs{\beta}(\Delta, i)^{-1} \right).
\nonumber
\end{align}
\end{lem}
\begin{proof}
The proof is in Appendix \ref{Sec:L5}.
\end{proof}

\begin{lem}
\label{lem:ineq}
When the conditions of Theorem \ref{thm:upperbound} hold, we have
\begin{align}
 \textup{tr}\left( \bs{\Phi}(\Delta, i) \bs{\beta}(\Delta, i)^{-1} \right) &\leq
\textup{tr}\left( \bs{\Phi}^{(u)}(\Delta, i) \bs{\beta}(\Delta, i)^{-1} \right)
\label{eq:gammabound1}
\\
  \textup{tr}\left( \bs{\Phi}^{(u)}(\Delta, i) \bs{\beta}(\Delta, i)^{-1} \right)
&\leq \bar{\gamma}^{(u)}(\Delta, i).
\label{eq:gammabound2}
\end{align}
\end{lem}

\begin{proof}
When the conditions of Theorem \ref{thm:upperbound} hold,
Lemma \ref{lem:beta} implies that
 $\bs{\Phi}(\Delta, i)\preceq \bs{\Phi}^{(u)}(\Delta, i)$
 and
$\mx{Z}(\Delta, i) \succeq \mx{Z}^{(u)}(\Delta, i)$.
Using the first relation and 
the Lemma \ref{lem:AB}
gives \eqref{eq:gammabound1},
while using the second relation and Lemma \ref{lem:AB} gives \eqref{eq:gammabound2}.
\end{proof}


\subsection{Useful Properties of the Upper Bounds on the Deterministic Equivalent \ac{SINR} and Overall System Spectral Efficiency}

Theorem \ref{thm:upperbound} is useful, because it establishes an upper bound,
denoted by $\bar{\gamma}^{(u)}(\Delta, i)$,
of the deterministic equivalent of the \ac{SINR},
$\bar{\gamma}(\Delta, i)$.

To use the $\bar{\gamma}^{(u)}(\Delta, i)$ upper bound for limiting the search space for an optimal $\bar{\gamma}(\Delta, i)$ in Section \ref{Sec:Alg}, we need the following properties of the upper bound.


\begin{prop}
\label{UpperB1}
The $\bar{\gamma}^{(u)}(\Delta, i)$ upper bound has the following properties:
$\partial \bar{\gamma}^{(u)}(\Delta, i) / \partial \rho(\Delta, i) \geq 0$ and $\rho(\Delta, i) \rightarrow 0 \Rightarrow \bar{\gamma}^{(u)}(\Delta, i) \rightarrow 0$.
\end{prop}
\begin{proof}
The proof is in Appendix \ref{Sec:P2}.
\end{proof}

Similarly, the SINR of user $k$ satisfies the inequality
$\bar{\gamma}_k(\Delta, i) \leq \bar{\gamma}_k^{(u)}(\Delta, i)$
where $\bar{\gamma}_k^{(u)}(\Delta, i)$ is defined in a similar way as $\bar{\gamma}_1^{(u)}(\Delta, i)$.
The $\bar{\gamma}_k^{(u)}(\Delta, i)$ upper bound is such that
$\partial \bar{\gamma}_k^{(u)}(\Delta, i) / \partial \rho_k(\Delta, i) \geq 0$ and $\rho_k(\Delta, i) \rightarrow 0 \Rightarrow \bar{\gamma}_k^{(u)}(\Delta, i) \rightarrow 0$.

Since our most important performance measure is the overall \ac{SE}, we are interested in establishing a corresponding upper bound on the overall \ac{SE} of the system.
To this end,
we introduce the related upper bound on the \ac{SE} of user $k$:
\begin{align}
\label{eq:SEku}
	\textup{SE}_k^{(u)}(\Delta) \triangleq \frac{ \sum_{i=1}^{\Delta}\log\big(1 + \bar{\gamma}^{(u)}_k(\Delta, i) \big) }{\Delta}.
\end{align}
and bound the aggregate average \ac{SE} of the \ac{MU-MIMO} system (c.f. \eqref{eq:multiopt}).
Notice that the denominator in $\textup{SE}_k^{(u)}$ is $\Delta$ while the denominator in $\textup{SE}_k$ is $\Delta + 1$.
This will be necessary for the monotonicity property in Proposition \ref{UpperB2}.
\begin{prop}
\label{UpperB2}
	\begin{align}
	\label{eq:SEupper}
	\textup{SE}^{(u)}(\Delta) \triangleq \sum_{k=1}^K \textup{SE}_k^{(u)}(\Delta)
	\geq \textup{SE}(\Delta),
	\end{align}
	and $\textup{SE}^{(u)}(\Delta)$ decreases with $\Delta$ and approaches $0$ when $\Delta$ approaches infinity.
\end{prop}

\begin{proof}
The proof is in Appendix \ref{Sec:P3}.
\end{proof}

\subsection{Summary}
This section first established an upper bound on the deterministic equivalent \ac{SINR} in Theorem \ref{thm:upperbound}. Next, Proposition \ref{UpperB1} and Proposition \ref{UpperB2} have stated some useful properties of this upper bound and a corresponding upper bound on the overall system spectral efficiency. Specifically, Proposition \ref{UpperB2} suggests that the upper bound on the spectral efficiency of the system is monotonically decreasing in $\Delta$ and tends to zero as $\Delta$ approaches infinity. As we will see in the next section, this property can be exploited to limit the search space for finding the optimal $\Delta$.

\section{A Heuristic Algorithm to Find the Optimum Pilot Power and Frame Size (Pilot Spacing)}
\label{Sec:Alg}
\subsection{A Heuristic Algorithm for Finding the Optimal $\Delta$}
In this section we build on the property of the system-wide spectral efficiency, as stated by Proposition \ref{UpperB2}, to develop a heuristic algorithm to find the optimal $\Delta$. While we cannot prove a convexity or non-convexity property of $\text{SE}(\Delta)$, we can utilize the fact that
$\text{SE}(\Delta) \leq \text{SE}^{(u)}(\Delta)$ as follows.
As Algorithm \ref{alg:Dynamic_Game}
 scans through the possible values of $\Delta$, it checks if the current best $\Delta$ (that is $\Delta_{\textup{opt}}$) is one less than the currently examined $\Delta$ (Line 17).
As it will be exemplified in Figure \ref{fig:6} in the numerical section, the key is to notice that the SE upper bound
determines the search space of the possible $\Delta$ values, where the associated SE can possibly exceed the currently found highest SE. Specifically, the search space can be limited to (Line 18):
\begin{align}
\Delta_{\textup{max}} &= \textup{SE}^{{(u)}^{-1}}(\textup{SE} _{\textup{$\Delta$}}),
\end{align}
where $\textup{SE}^{{(u)}^{-1}}$ denotes the inverse function of $\textup{SE}^{(u)}(.)$ and
$\textup{SE}_\Delta \triangleq \textup{SE}(\Delta)$ as calculated in \eqref{eq:multiopt}.

{\small
\begin{algorithm}[t!]
\DontPrintSemicolon
\caption{Optimum frame size algorithm using an SE upper bound}
\label{alg:Dynamic_Game}
\KwIn{
$\mx{Q}$, 
$\mx{C}$, $\mx{\Sigma}$, $\alpha^2$, $P_{\text{tot}}$ 
}
$\textup{SE}_{1}=\textup{SE}(1)$ using \eqref{eq:multiopt}, $\Delta_\text{max}={\textup{SE}^{(u)}}^{-1}(\textup{SE}_{1})$ \\
$\Delta=1$, $\Delta_{\textup{opt}} = \Delta_\text{max}$,
$\textup{SE}_{\textup{opt}}=\textup{SE}(\Delta_{\textup{opt}})$ using \eqref{eq:multiopt} \\
\While{$\Delta < \Delta_\textup{max}$\hspace{2mm}}{
\For{$k = 1 \ldots K$}{
\For{$i = 1 \ldots \Delta$\hspace{2mm}}{
Calculate $\mx{R}_k(i), \mx{R}_k(\Delta+1),$ \\
~~~$\mx{R}_k(\Delta+1 \pm i),\mx{R}_k(2\Delta+2)$ using \eqref{eq:autocorrk}\\
Calculate $\mx{E}_k(\Delta,i)$ using \eqref{eq:E}\\
Calculate $\mx{Z}_k(\Delta,i)$ using \eqref{eq:Z}\\
Calculate $\mx{\Phi}_k(\Delta,i)$ using \eqref{eq:Phi}\\
Calculate $\bs{\beta}_k(\Delta,i)$ using \eqref{eq:beta}\\
Calculate $\bar \gamma_k(\Delta,i)$ using
\eqref{eq:hoydis1}\\
Calculate $\text{SE}_k(\Delta,i)$ using \eqref{eq:SE}\\}}
$\textup{SE}_\Delta=\textup{SE}(\Delta)$ using \eqref{eq:multiopt} \\
\If{$\textup{SE}_\Delta > \textup{SE}_{\textup{opt}}$}{
    $\Delta_{\textup{opt}} = \Delta$, $\textup{SE}_{\textup{opt}}=\textup{SE}_\Delta$
    }
	\If{$\Delta_{\textup{opt}} = \Delta-1 $}
    {
    $\Delta_\text{max}={\textup{SE}^{(u)}}^{-1}(\textup{SE}_{\Delta})$
    }
$\Delta=\Delta+1$
}
\KwOut{$\Delta_\textup{opt}$}
\medskip
\end{algorithm}
}

\subsection{The Case of Independent and Identical Channel Coefficients}
\label{Sec:IID}

In the special case where the elements of the vector $\mx{h}(i)$ are independent stochastically identical stochastic processes, the covariance matrices become real multiples of the identity matrix
$\mx{C} \triangleq c \mx{I}_{N_r}$, $\bs{\Sigma} = s \mx{I}_{N_r}$, $\mx{R}(i) = r(i) \mx{I}_{N_r}$, $\mx{Z}(i)= z(i) \mx{I}_{N_r}$, $\bs{\Phi}(i) = \phi(i)\mx{I}_{N_r}$,
$\bs{\beta}(i) = \beta(i) \mx{I}_{N_r}$,
further more $\mx{E}(i) = \mx{e}(i) \otimes \mx{I}_{N_r}$,
with:
\begin{align}
s &\triangleq \frac{\sigma_p^2}{\alpha^2P_p \tau_p},\\
\label{eq:Rdiag}
r(i) &\triangleq ce^{q^* i},
\end{align}
\begin{align}
\label{eq:Ediag}
 \mx{e}(i) &\triangleq
\begin{bmatrix}r(\Delta+1+i)&r(i)&r(\Delta+1-i)\end{bmatrix}
\nonumber\\
&~~\cdot\begin{bmatrix}
c+s & r(\Delta+1) &r(2\Delta+2) \\
r^H(\Delta+1) &c+s &r(\Delta+1) \\
r^H(2\Delta+2) & r^H(\Delta+1) & c+s\\
\end{bmatrix}^{-1}
\end{align}
\begin{align}
\label{eq:Zdiag}
 z(i) &\triangleq
\left(c - \mx{e}(i)  \begin{bmatrix}
r^H(\Delta+1+i) \\
r^H(i) \\
r^H(\Delta+1-i) \\
\end{bmatrix}
\right),
\end{align}
\begin{align}
\label{eq:Phidiag}
 \phi(i)  &\triangleq \alpha^2P(i)(c-z(i)),
\end{align}
\begin{align}
\label{eq:Betadiag}
\beta(i) &\triangleq  \left(\sum_{k=1}^K \alpha_k^2 P_k  z_k(i) + \sigma^2_d\right).
\end{align}

In this special case, calculating the deterministic equivalent of the \ac{SINR} by Proposition \ref{Prop:SINR} simplifies to solving a set of scalar equations as stated in the following corollary.

\begin{cor}
\label{Cor:DetEq}
In this special case, the deterministic equivalent of the \textup{SINR} in slot $i$,
$\bar{\gamma}(i)$, can be obtained as the solution of the scalar equation
\begin{align}
\beta(i)  &=
\frac{N_r \phi(i)}{\bar{\gamma}(i)} - \sum_{k = 2}^K \frac{\phi_k(i)}{1 + \frac{\bar{\gamma}(i)\phi_k(i)}{\phi(i)}}.
\label{eq:diag_hoydis1}
\end{align}
\end{cor}

\begin{proof}
Since the matrices $\bs{\Phi}_k(i)$ and $\mx{Z}_k(i)$ are constant multiple of identity matrices, \eqref{eq:hoydis3} can then be rewritten as
\begin{align}
\delta_k(i) &= N_r \phi_k(i) \left( \sum_{l = 2}^{K} \frac{\phi_l(i)}{1 + \delta_l(i)} +
\beta(i) \right)^{-1}
\label{eq:diag_hoydis}
\end{align}
for $k = 1, \ldots, K$.
Using $\bar{\gamma}(i) = \delta_1(i)$ and comparing \eqref{eq:diag_hoydis} for different values of $k$ we get
\begin{align}
\delta_k(i) = \frac{ \phi_k(i)  }{  \phi_1(i)  } \delta_1(i) = \frac{ \phi_k(i)  }{  \phi_1(i) } \bar{\gamma}(i).
\label{eq:deltarealtion}
\end{align}
Substituting the rightmost expression of \eqref{eq:deltarealtion} into \eqref{eq:diag_hoydis} with $k = 1$ and rearranging gives the corollary.
\end{proof}

Notice that calculations inside the inner for loop of Algorithm \ref{alg:Dynamic_Game}, that is the
calculations in Lines 6-13 can be substituted by equations
\eqref{eq:Rdiag}, \eqref{eq:Ediag}, \eqref{eq:Zdiag}, \eqref{eq:Phidiag} and \eqref{eq:Betadiag}.

\section{Numerical Results}
\label{Sec:Num}

\begin{table}[ht]
\caption{System Parameters}
\vspace{1mm}
\label{tab:params}
\footnotesize
\begin{tabularx}{\columnwidth}{|X|X|}
		\hline
		\hline
		\textbf{Parameter}                     & \textbf{Value} \\
		\hline
		\hline
		Number of receive antennas at the \ac{BS} antennas  & $N_r=10, 100$  \\ \hline
		Path loss of the tagged MS               & $\alpha=90$ dB \\ \hline
        Frame size                              & $\Delta=2 \dots 50$ \\ \hline
		Pilot and data power levels             & $P_p=50...125$ mW; $P=125$ mW \\ \hline
        MIMO receivers                         & MMSE receiver given by \eqref{eq:Gstar2} \\ \hline
        Channel estimation                      & MMSE channel estimation given by Lemma \ref{lem:mmsechannel} \\ \hline
        Maximum Doppler frequency               & $f_D=50, 500, 1500$ Hz \\ \hline
        Slot duration ($T$)                     & $32\mu$s \\ \hline
        Number of users                         & $K=2$ \\ \hline
		\hline
\end{tabularx}
\end{table}

In this section, we consider a single cell of a \ac{MU-MIMO} ($K=2$) system with $N_r=10$ and $N_r=100$ receive antennas,
in which the wireless channel between the served \ac{MS} and the \ac{BS} is Rayleigh fading according to \eqref{eq:RayleighApprox}, which we approximate with \eqref{eq:r1}.

The \ac{MU-MIMO} case with greater number of users ($K>2$) gives similar results
albeit with somewhat lower \ac{SINR} values from the point of view of the tagged user.
The \ac{BS} estimates the state of the wireless channel based on
the properly (i.e. $\Delta \times T$) spaced the pilot signals using \ac{MMSE} channel estimation and interpolation according to Lemma \ref{lem:mmsechannel}, and uses \ac{MMSE}
symbol estimation employing the optimal linear receiver $\mx{G}^\star(i T)$ in each slot as given in \eqref{eq:Gstar2}.
Specifically, except for the results shown in Figure \ref{fig:7},
in each time slot $i=1 \dots \Delta$, the \ac{BS} uses one pilot signal transmitted by the \ac{MS} at the
beginning of the frame at time instance $i=0$ and one pilot sent at the beginning of the next frame at time
instance $i=\Delta+1$. We refer to these two pilot signals as sent "before" and "after" time slot $i$.
In practice, the \ac{BS} can store the received data symbols until it receives the pilot signal in slot
$i=\Delta+1$ before using an \ac{MMSE} interpolation of the channel states between $i=0$ and $i=\Delta+1$.
Furthermore,
we will assume that the \ac{BS} estimates
perfectly the autocorrelation function of the channel, including the associated maximum Doppler frequency
and, consequently, the characterizing zeroth order Bessel function.
The most important system parameters are listed in Table \ref{tab:params}.
Here we assume that the slot duration ($T$) corresponds to a symbol duration in
5G \ac{OFDM} systems using 122 MHz clock frequency, which can be used up to 20 GHz
carrier frequencies \cite{Zaidi:16}.
Note that the numerical results presented below are obtained by using the 
results on the deterministic equivalent of the \ac{SINR} and the corresponding average spectral efficiency.

\begin{figure}[t]
\begin{center}
\includegraphics[width=\hsize]{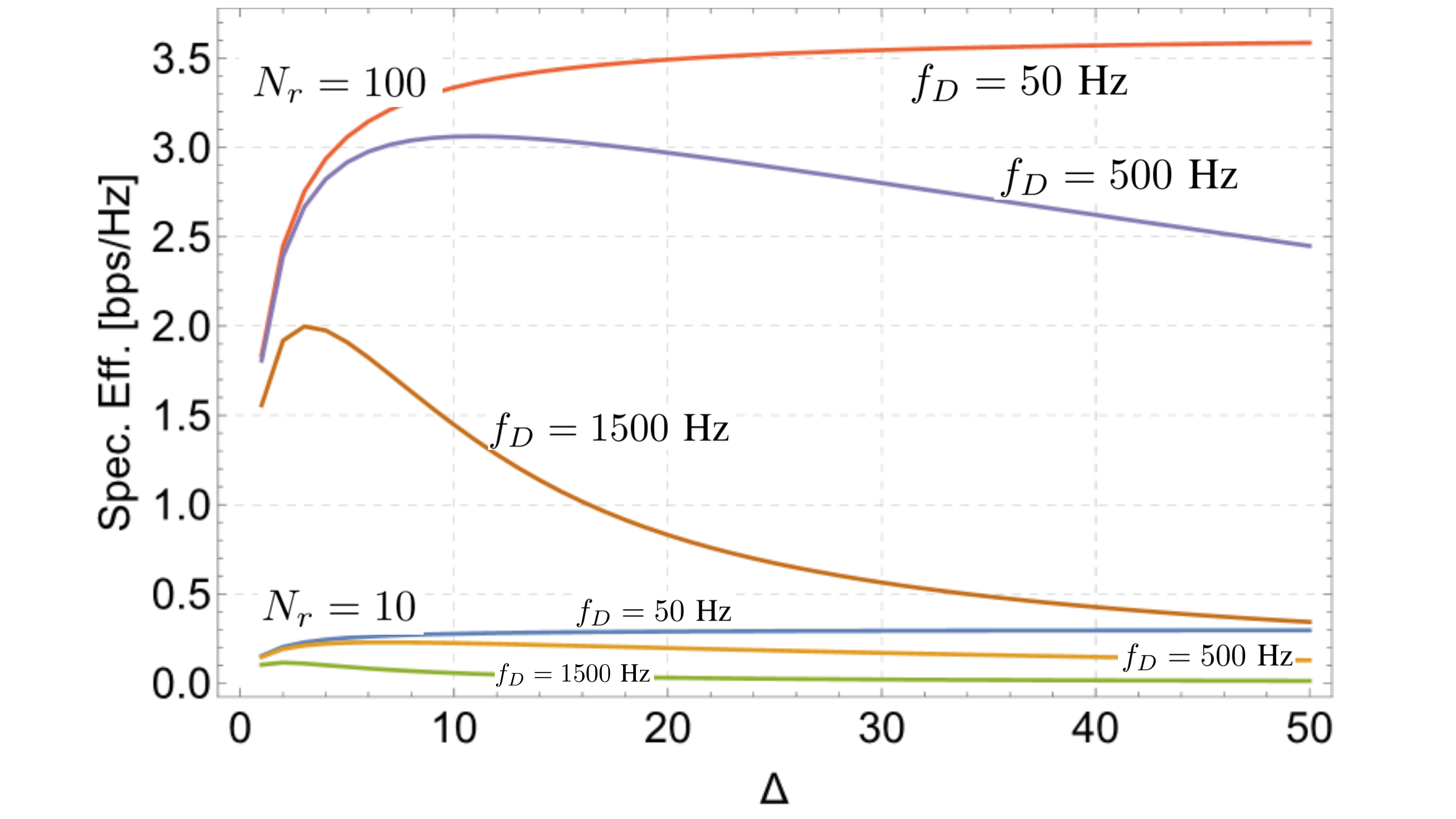}
\caption{
Spectral efficiency as a function of frame size ($\Delta$) with maximum Doppler frequency
$f_D=50, 500, 1500$ Hz with $N_r=10$ (lower three curves) and $N_r=100$ (upper three curves).
At higher maximum Doppler frequency, the optimum frame size is smaller than at low Doppler
frequency.
}
\label{fig:1}
\end{center}
\end{figure}

Figure \ref{fig:1} shows the achieved spectral efficiency averaged over the data slots $i=1\dots\Delta$,
that is averaged over the data slots of a frame of size $\Delta+1$.
Short frames imply that the pilot overhead
is relatively large, which results in poor spectral efficiency.
On the other hand, too large frames (that is when $\Delta$
is too large) make the channel estimation quality in the "middle" time slots poor, since for these time slots both
available channel estimates $\mx{\hat h}(0)$ and $\mx{\hat h}(\Delta+1)$ convey little useful information, especially
at high Doppler frequencies when the channel ages rapidly.
Indeed, as seen in Figure \ref{fig:1}, the frame size
has a large impact on the achievable spectral efficiency, suggesting that the optimum frame size depends critically on
the Doppler frequency.
As we can see, the spectral efficiency as a function of the frame size is in general neither monotone nor concave,
and is hence hard to optimize.

\begin{figure}[t]
\begin{center}
\includegraphics[width=\hsize]{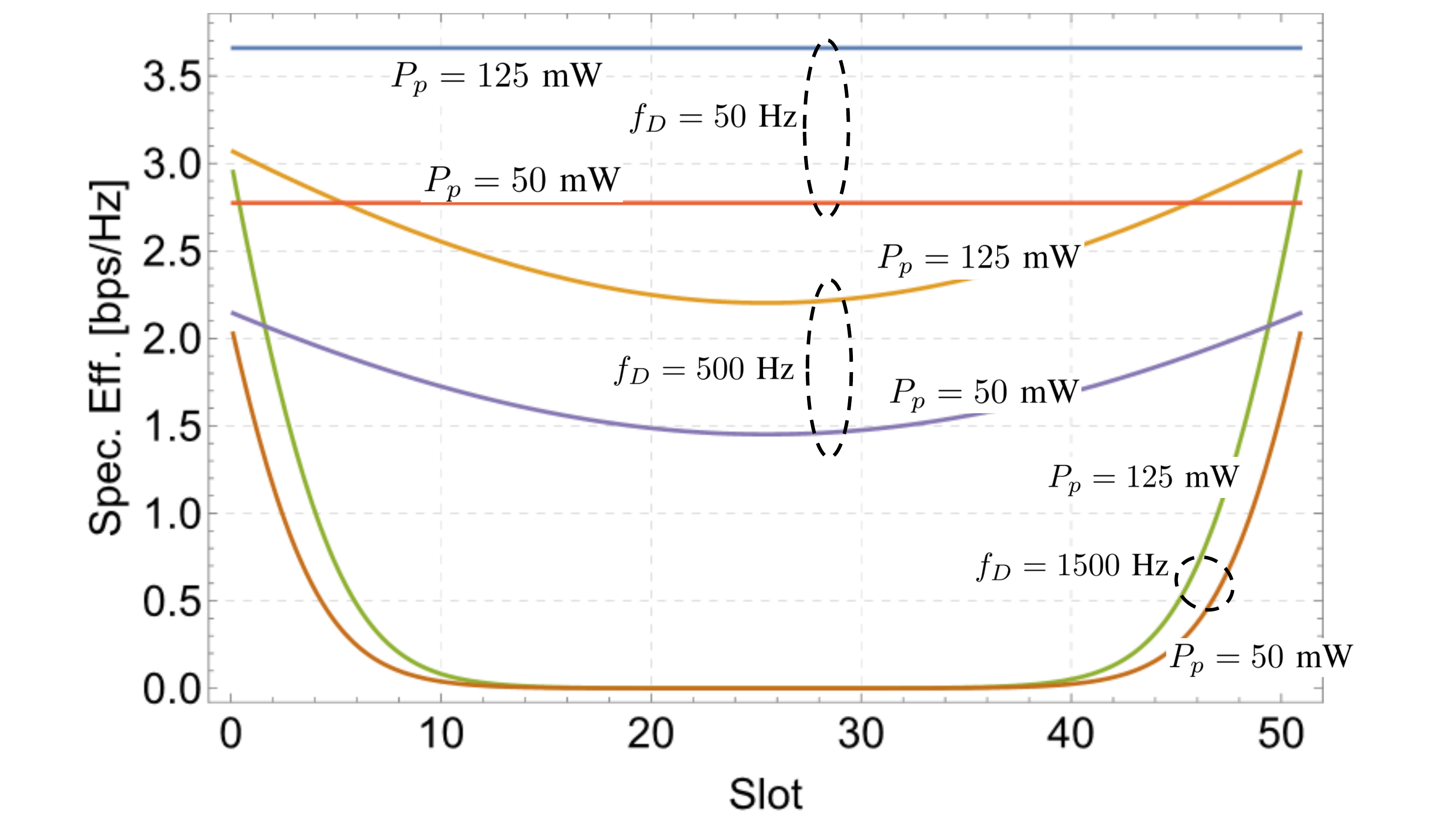}
\caption{
Spectral efficiency for each data slot $i=1 \dots \Delta$ when the frame size is kept fixed ($\Delta=50$).
At high maximum Doppler frequency, the spectral efficiency is low at the "middle" slot,
while at low maximum Doppler the spectral efficiency reaches its maximum at the middle slots.
}
\label{fig:2}
\end{center}
\end{figure}

Figure \ref{fig:2} shows the spectral efficiency for
each data slot $i=1 \dots \Delta$ within a frame of size $\Delta=50$.
At lower Doppler frequencies, that is when the channel fades relatively slowly,
the channel state information acquisition in the middle slots benefits from using the estimates at $i=0$ and $i=\Delta+1$,
and making an \ac{MMSE} interpolation of the channel coefficients as proposed in Lemma \ref{lem:mmsechannel}.	
However, at a high Doppler frequency, the channel state in the middle data slots are weakly
correlated with the channel estimates $\mx{\hat h}(0)$ and $\mx{\hat h}(\Delta+1)$,
which makes the \ac{MMSE} channel estimation error in Corollary \ref{cor:rmmse} large.
This insight suggests that in such cases, the optimum frame size
is much less than when the Doppler frequency is low.

\begin{figure}[t]
\begin{center}
\includegraphics[width=\hsize]{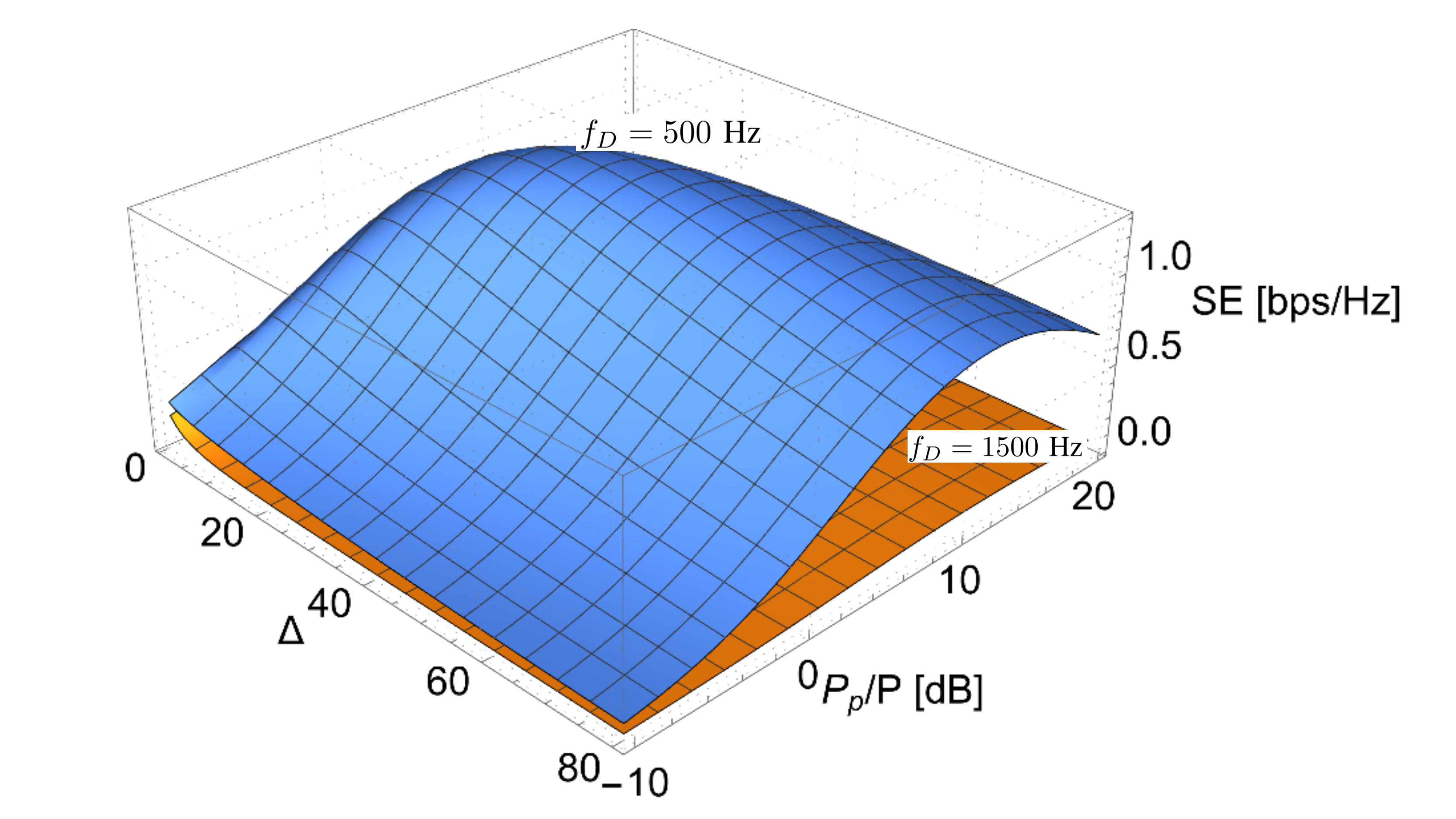}
\caption{
Spectral efficiency as a function of the pilot/data power ratio and the frame size
for maximum Doppler frequency $f_D=500$ Hz and $f_D=1500$ Hz when $N_r=10$.
In both cases, the spectral efficiency depends heavily on the employed pilot power
and pilot spacing (frame size).
}
\label{fig:3}
\end{center}
\end{figure}

The average spectral efficiency as a function of the pilot/data power ratio and the
frame size is shown in Figure \ref{fig:3}. This figure clearly shows that setting the
proper frame size and tuning the pilot/data power ratio are both important to maximize
the average spectral efficiency of the system. The optimal frame size and power
configuration are different for different Doppler frequencies, which in turn emphasizes
the importance of accurate Doppler frequency estimates.

\begin{figure}[t]
\begin{center}
\includegraphics[width=\hsize]{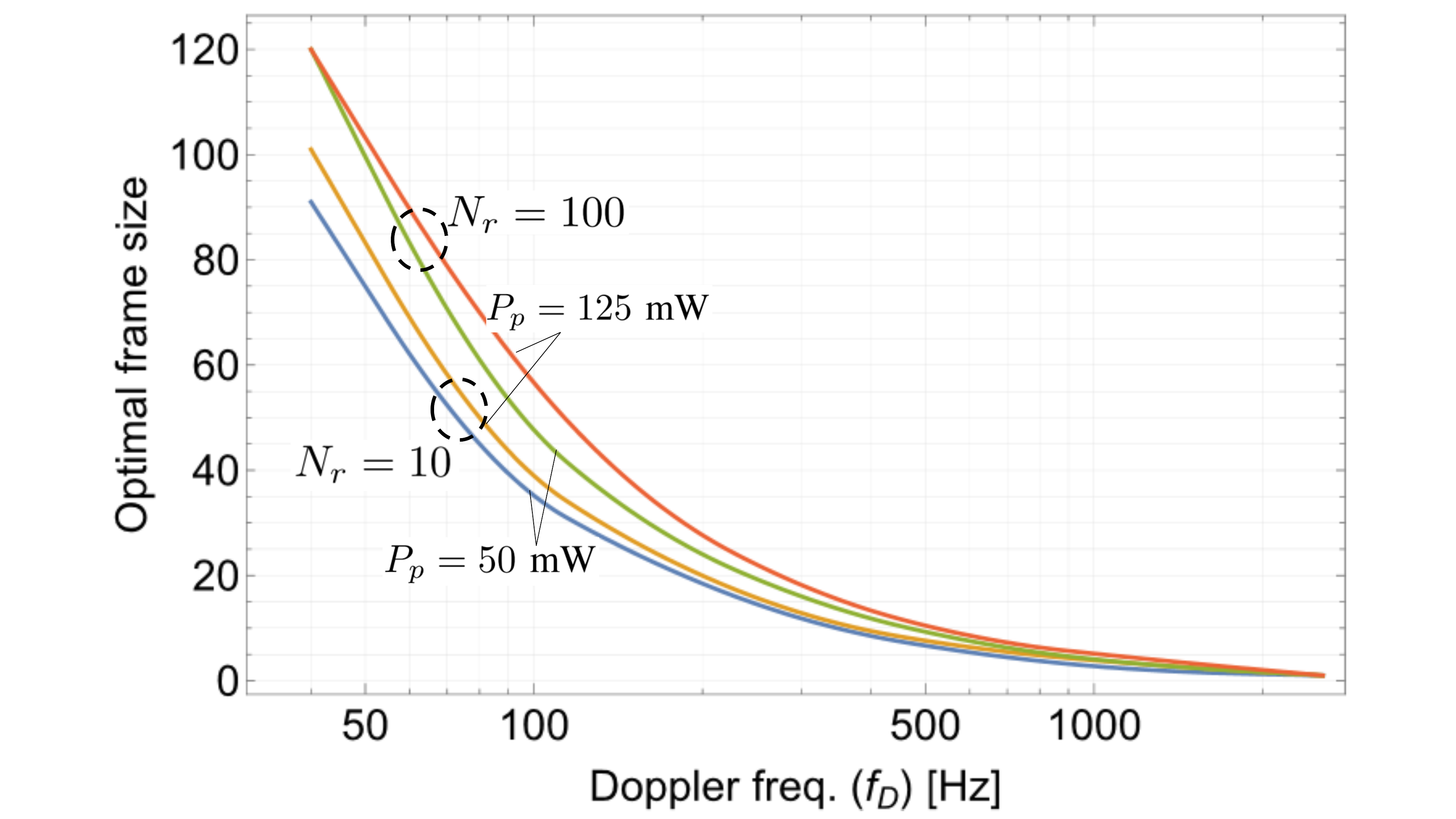}
\caption{
Optimal frame size as a function of the maximum Doppler frequency for different values of
the employed pilot power ($P_p=50$ mW and $P_p=125$ mW) when the BS is equipped with $N_r=10$
and $N_r=100$ receive antennas.
}
\label{fig:4}
\end{center}
\end{figure}

The optimal frame size as a function of the maximum Doppler frequency is shown in Figure \ref{fig:4}.
The optimal frame size decreases rapidly, as the Doppler effect increases. As this figure shows,
a much larger frame size is optimal when the number of antennas is high and the \ac{MS} uses high
pilot power to achieve a high pilot \ac{SNR}.

\begin{figure}[t]
\begin{center}
\includegraphics[width=\hsize]{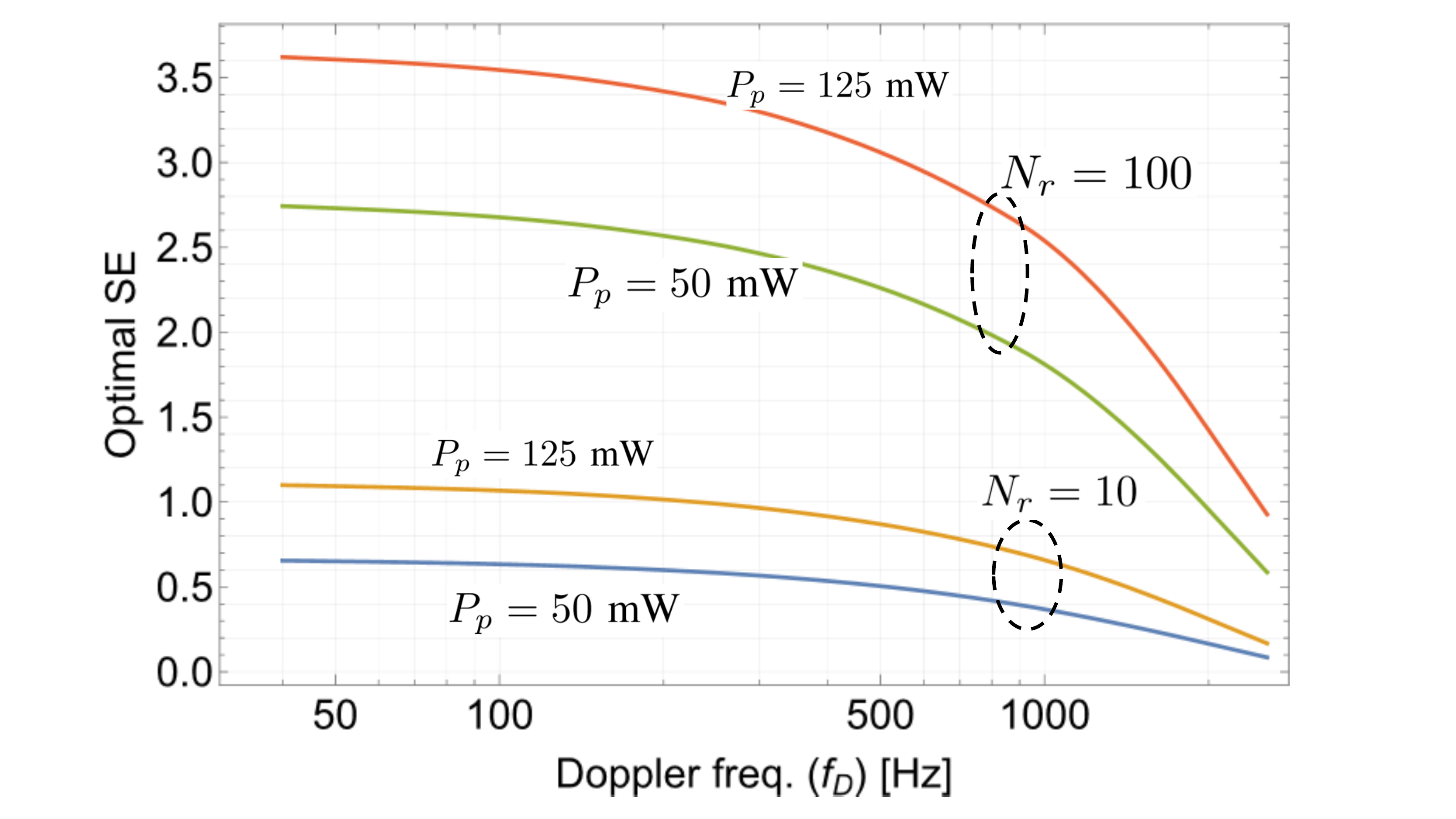}
\caption{
Optimal spectral efficiency as a function of the maximum Doppler frequency, that is the spectral
efficiency when using the optimal frame size as shown in Figure \ref{fig:4}.
}
\label{fig:5}
\end{center}
\end{figure}

Figure \ref{fig:5} shows the achieved spectral efficiency when the frame size is set optimally,
as a function of the maximum Doppler frequency. At $f_D=500$ Hz, for example, when the optimal
frame size is 8 (see also Figures \ref{fig:1} and \ref{fig:4}), the achieved spectral efficiency
when using $N_r=10$ antennas is a bit below 1 bps/Hz. We can see that setting the optimal frame
size is indeed important, because it helps to make the achievable spectral efficiency quite
robust with respect to even a significant increase in the Doppler frequency.

\begin{figure}[t]
\begin{center}
\includegraphics[width=\hsize]{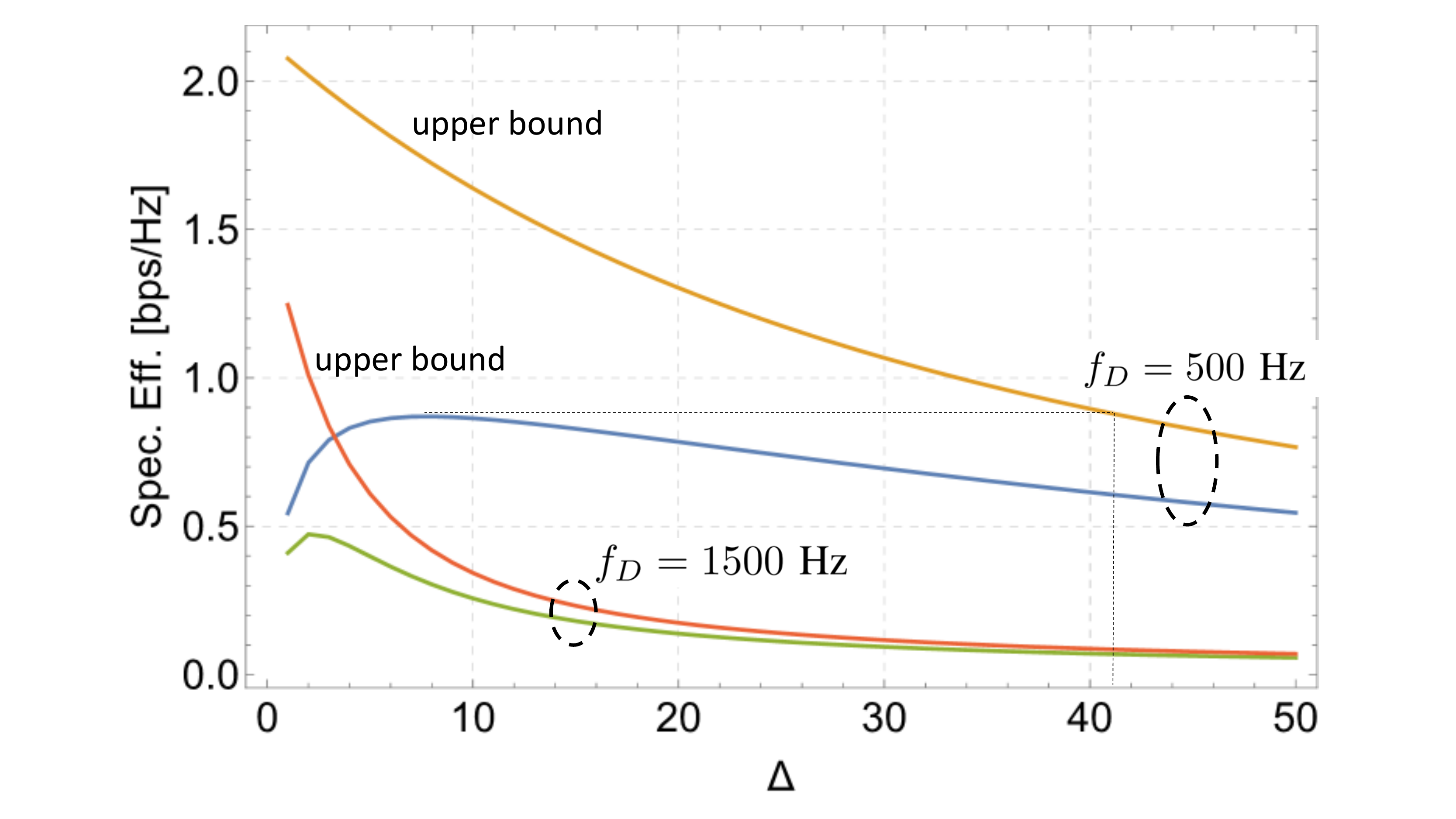}
\caption{
Upper bounding the achievable spectral efficiency as a function of the frame size ($\Delta$)
at $f_D=500$ Hz and $f_D=1500$ Hz. Note that the upper bound is monotonically decreasing, which
helps to limit the search space for the optimum frame size.
}
\label{fig:6}
\end{center}
\end{figure}

Figure \ref{fig:6} illustrates the upper bounds on spectral efficiency as a function of the
frame size for different Doppler frequencies. Recall from Figure \ref{fig:1} that the spectral
efficiency of the system is a non-concave function of the frame size. Therefore, limiting the
possible frame sizes that can optimize spectral efficiency is useful, which can be achieved
by the upper bounds shown in the figure. Since the upper bound is monotonically decreasing,
finding a point of the spectral efficiency curve (see the curve marked with $f_D=500$ Hz
and its upper bounding curve) with a negative derivative helps to find the range of possible
frame sizes that maximize spectral efficiency. For $f_D=500$ Hz, as illustrated in the figure,
larger frame sizes than $\Delta=41$ would lead to a lower upper bound
than the spectral efficiency achieved at $\Delta=7$.
Therefore, when searching for the optimal $\Delta$, once we found that the spectral efficiency
at $\Delta=8$ is less than at $\Delta=7$ (negative derivative), the search space is limited to $(7,41)$.

\begin{figure}[t]
\begin{center}
\includegraphics[width=\hsize]{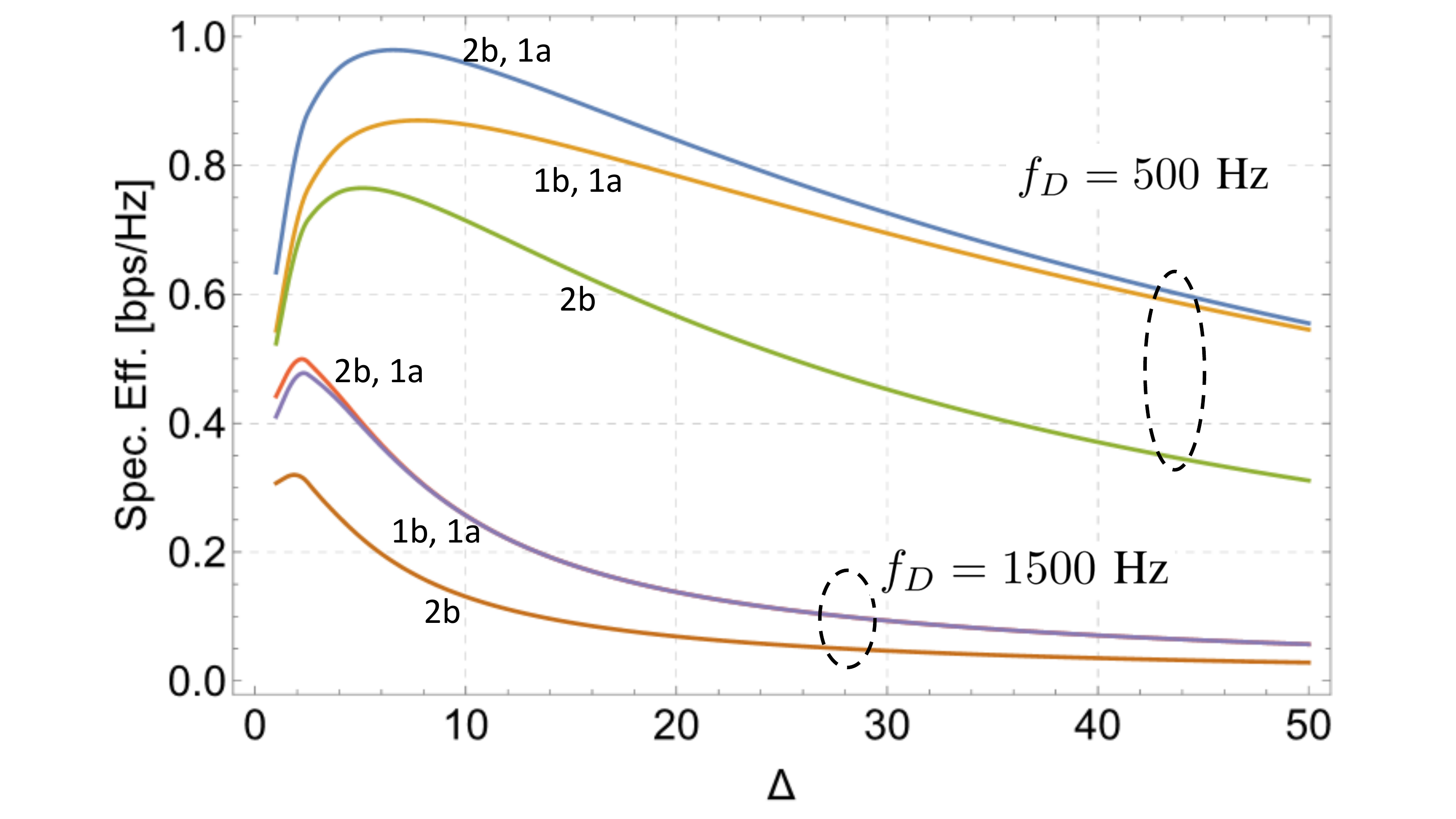}
\caption{
Spectral efficiency in each time slot for $f_D=500$ Hz and $f_D=1500$ Hz, when using
1 or 2 pilot symbols preceding that time slot and 0 or 1 pilot symbols after that
time slot for channel estimation. Three combinations of these channel estimation schemes
are denoted as "2b, 1a", "1b, 1a" and "2b", where "b" refers to utilizing the pilot symbols
sent before and "a" refers to utilizing the pilot symbol sent after time slot $i$.
}
\label{fig:7}
\end{center}
\end{figure}

Figure \ref{fig:7} compares the average spectral efficiency when the system uses different
number of pilot signals to estimate the channel state for each data slot within the frame.
Specifically, three schemes are compared:
\begin{itemize}
\item
2 before, 1 after (2b, 1a): Three channel estimates using the pilot signals at the beginning
of the current frame and the preceding frame and at the end of the current frame are used
to interpolate the channel state at every data slot in the current frame.
\item
1 before, 1 after (1b, 1a): The two neighboring pilot signals (that is in the beginning and at the end
of the current frame) are used.
\item
2 before (2b): The pilot signals at the beginning of the current and preceding frames are used.
This scheme has an advantage over the previous schemes in that decoding the received data symbols
is possible "on the fly" without having to await the upcoming pilot signal at the end of the
current frame.
\end{itemize}

Notice that the "1b, 1a" scheme outperforms the "2b" scheme, because the channel estimation
instances are closer to the data transmission instance in time. Furthermore, the "2b, 1a"
scheme further improves the \ac{SE} performance, although this improvement over the
"1b, 1a" scheme is marginal. More importantly, we can observe that the optimal pilot
spacing is similar in these three schemes, but depends heavily on the Doppler frequency.

\begin{figure}[t]
\begin{center}
\includegraphics[width=\hsize]{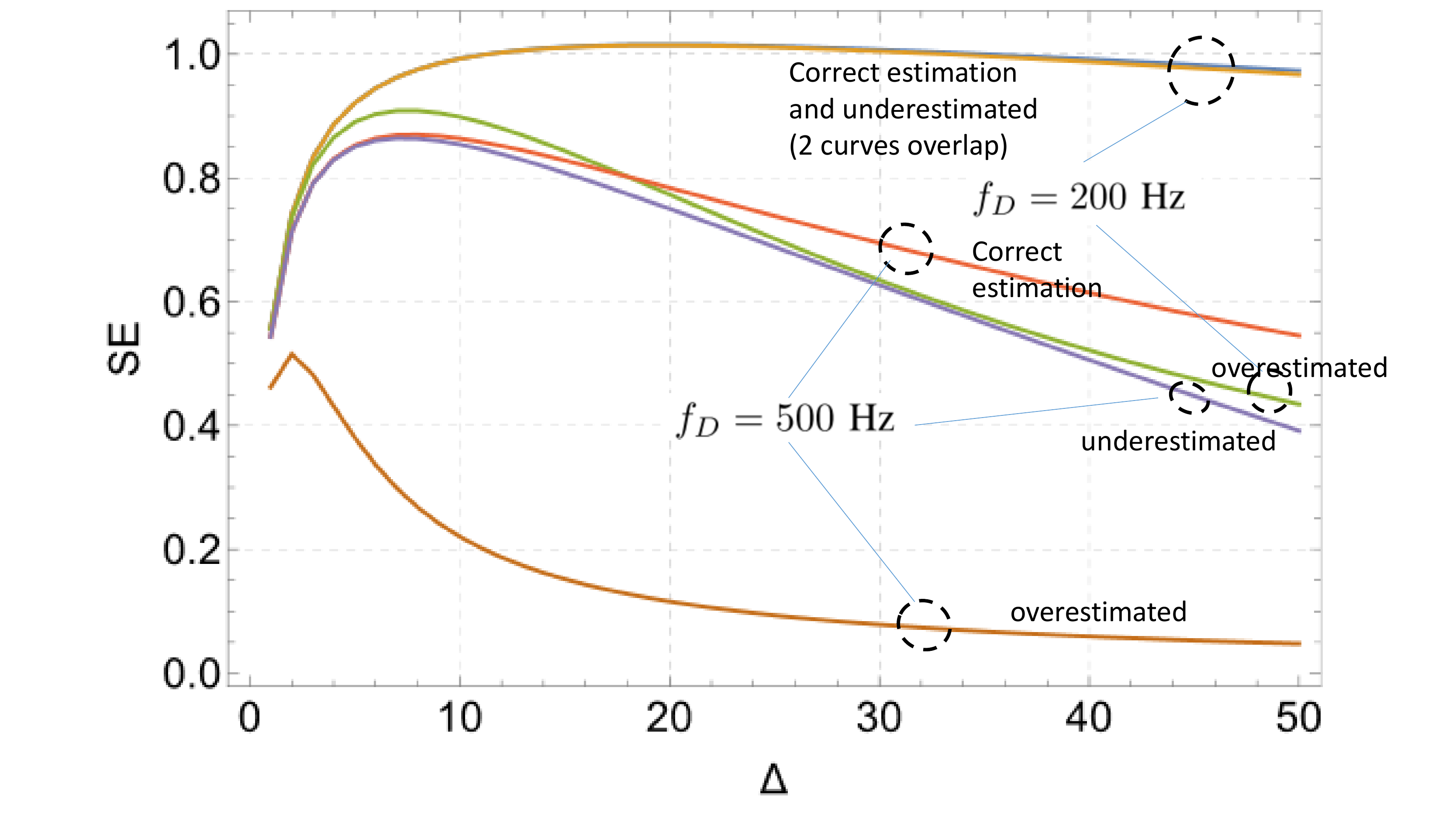}
\caption{
Spectral efficiency as a function of the frame size $\Delta$ when the receiver under or overestimates the actual Doppler frequency of the channel ($f_D=200$ Hz and $f_D=500$ Hz). Overestimating the actual Doppler frequency causes significant spectral efficiency degradation for most frame sizes.
}
\label{fig:8}
\end{center}
\end{figure}

Finally, Figure \ref{fig:8} examines the negative impact of Doppler frequency estimation
errors when
the Doppler frequency of the channel is under or overestimated.
The figure shows the spectral efficiency as a function of the frame size for the cases when
$f_D=200$ Hz and $f_D=500$ Hz. For both cases, the Doppler frequency is either correctly estimated or overestimated (to $5f_D$) or underestimated (to $0.2 f_D$).
On the one hand, this figure clearly illustrates the performance degradation
in terms of average spectral efficiency when the receiver underestimates or overestimates
the maximum Doppler frequency.
On the other hand, when using the optimal frame size, the spectral efficiency performance
of these schemes are rather similar in most cases.

\section{Conclusions}
\label{Sec:Conc}
This paper investigated the fundamental trade-off between using resources in the time domain for pilot signals and data signals in the uplink of \ac{MU-MIMO} systems operating over fast fading wireless channels that age between subsequent pilot signals. While previous works indicated that when the autocorrelation coefficient between subsequent channel realization instances in discrete time is high, both the channel estimation and the \ac{MU-MIMO} receiver can take advantage of the memoryful property of the channel in the time domain. However, previous works do not answer the question how often the channel should be observed and estimated such that the subsequent channel samples are sufficiently correlated while
taking into account that pilot signals do not carry information bearing symbols and degrade the overall spectral efficiency.
To find the optimal pilot spacing, we first established the deterministic equivalent of the achievable \ac{SINR} and the associated overall spectral efficiency of the \ac{MU-MIMO} system. We then used some useful properties of an upper bound
of this spectral efficiency, which allowed us to limit the search space for the optimal pilot spacing ($\Delta$).
The numerical results indicate that the optimal pilot spacing is sensitive to the Doppler frequency of the channel and that proper pilot spacing has a significant impact on the achievable spectral efficiency.

\appendices
\section{Proof of Lemma \ref{lem:beta}}
\label{Sec:L4}
\begin{proof} 
Notice that
\begin{align}
\label{eq:MDelta}
\mx{M}(\Delta,i) &=
\underbrace{\begin{bmatrix}
1 & e^{2qi} & e^{4qi}  \\
{(e^{2qi})}^* & 1 & e^{2qi} \\
{(e^{4qi})}^* &  {(e^{2qi})}^* & 1\\
\end{bmatrix}}_{\triangleq \mx{M}_3(\Delta,i)}
 \otimes \mx{C},
\end{align}
and the eigenvalues of $\mx{M}_3(\Delta,i)$ are:
\begin{align*}
\lambda_1(i) &= 1 - a^{2i}; \\
\lambda_2(i) &= \frac{1}{2}\left(2+a^{2i}-a^i\sqrt{8+a^{2i}}\right);\\
\lambda_3(i) &= \frac{1}{2}\left(2+a^{2i}+a^i\sqrt{8+a^{2i}}\right),
\end{align*}
where $0<a<1$.
For all $i\geq1$ the smallest eigenvalue is $\lambda_2(i)$, which monotone increases with $i$. That is, $\min_{i\geq 1,j\in\{1,2,3\}} \lambda_j(i) = \lambda_2(1)$.

Let
\begin{align}
\mx{M}^{(u)}(\Delta) &\triangleq
\underbrace{\begin{bmatrix}
\eta & 0 & 0  \\
0 & \eta & 0  \\
0 &  0 & \eta \\
\end{bmatrix}}_{\triangleq \mx{M}^{(u)}_3(\Delta)}
 \otimes \mx{C}.
\end{align}

When $\eta<\lambda_2(1)$ according to \eqref{eq:etacond}, we have
\begin{align}
\label{lem3:13}
\mx{M}_3^{(u)}(\Delta) \preceq \mx{M}_3(\Delta).
\end{align}
Utilizing that
the spectrum of a Kronecker product $\sigma(\mx{A}\otimes \mx{B})$ is \cite{Horn:91}
\begin{align}
    \sigma(\mx{A}\otimes\mx{B}) &=
    \{\,
    \mu_A\mu_B \mid \mu_A \in \sigma(\mx{A}), \mu_B \in
    \sigma(\mx{B}) \,\},
\end{align}
for $\forall i\geq 1$, we further have
\begin{align}
\label{lem3:1}
\mx{M}^{(u)}(\Delta) \preceq \mx{M}(\Delta,i),
\end{align}
which implies
\begin{align}\label{eq:Mu}
\big(\mx{M}^{(u)}(\Delta)+\mx{\Sigma}_3\big)^{-1} \succeq  \big(\mx{M}(\Delta,i)+\mx{\Sigma}_3\big)^{-1},
\end{align}
according to \eqref{eq:ABinv}.
The statement of the lemma 
comes from \eqref{eq:Mu} using \eqref{eq:AB3},
$\mx{M}^{(u)}(\Delta)=\eta \mx{I}_{3} \otimes \mx{C}$,
and noting that
\begin{align}
\nonumber
&\mx{E}(\Delta, i) \big(\mx{M}^{(u)}(\Delta)+\mx{\Sigma}_3\big)^{-1} \mx{E}(\Delta, i)^H
\\\nonumber &= \mx{E}(\Delta, i) \begin{bmatrix}
\eta \mx{C} + \bs{\Sigma} & 0 & 0  \\
0 & \eta \mx{C} + \bs{\Sigma} & 0  \\
0 &  0 & \eta \mx{C} + \bs{\Sigma} \\
\end{bmatrix}^{-1} \mx{E}(\Delta, i)^H
\\\nonumber &=
\mx{R}(\Delta \!+\! 1 \!+\! i)
(\eta \mx{C} + \bs{\Sigma})^{-1}
\mx{R}(\Delta \!+\! 1 \!+\! i)^H
\\\nonumber &~~~+
\mx{R}(i)(\eta \mx{C} + \bs{\Sigma})^{-1}
\mx{R}(i)^H
\\\nonumber &~~~+
\mx{R}(\Delta \!+\! 1 \!-\! i)
(\eta \mx{C} + \bs{\Sigma})^{-1}
\mx{R}(\Delta \!+\! 1 \!-\! i)^H
\\&= \rho(\Delta, i) \mx{C} \left( \eta \mx{C} + \bs{\Sigma} \right)^{-1} \mx{C},
\end{align}
where $\mx{R}(i)$ and $\rho(\Delta, i)$ are defined in \eqref{eq:r1} and \eqref{eq:rho}.
\end{proof}

\section{Proof of Lemma \ref{lem:det}}
\label{Sec:L5}
\begin{proof}
\begin{align*}
& \bar{\gamma}(\Delta, i)
\\&= \mathds{E}\left( \textup{tr}\left( \bs{\Phi}(\Delta, i)  \left( \sum_{k=2}^K  \mx{b}_l(\Delta, i) \mx{b}_l^H(\Delta, i)  +\bs{\beta}(\Delta, i) \right)^{-1}  \right) \right)
\end{align*}
\begin{align*}
&=  \int\limits_{\mx{v}_2\in\mathbb{R}^{N_r}} \!\!\!\ldots\!\!\!  \int\limits_{\mx{v}_K\in\mathbb{R}^{N_r}}  \prod_{k=2}^K Pr(\mx{b}_l(\Delta, i)=\mx{v}_l)
\\& ~~~~~~\cdot \textup{tr}\left( \bs{\Phi}(\Delta, i)  \left( \sum_{k=2}^K  \mx{v}_l \mx{v}_l^H  +\bs{\beta}(\Delta, i) \right)^{-1} \right)
 d\mx{v}_K \ldots d\mx{v}_2
\end{align*}
\begin{align*}
& \leq \int\limits_{\mx{v}_2\in\mathbb{R}^{N_r}} \int\limits_{\mx{v}_K\in\mathbb{R}^{N_r}}  \prod_{k=2}^K Pr(\mx{b}_l(\Delta, i)=\mx{v}_l)
\\& ~~~~~~~~~~~~~~~~~~\cdot
\textup{tr}\left( \bs{\Phi}(\Delta, i)   \bs{\beta}(\Delta, i)^{-1} \right) d\mx{v}_K \ldots d\mx{v}_2
\\&= \textup{tr}\left( \bs{\Phi}(\Delta, i) \bs{\beta}(\Delta, i)^{-1} \right),
\nonumber
\end{align*}
where we used that $\sum_{l = 2}^K  \mx{v}_l \mx{v}_l^H$ is a positive definite matrix,
$\sum_{l = 2}^K  \mx{v}_l \mx{v}_l^H +\bs{\beta}(\Delta, i) \succeq \bs{\beta}(\Delta, i)$
and Lemma \ref{lem:AB}.
\end{proof}

\section{Proof of Proposition \ref{UpperB1}}
\label{Sec:P2}
\begin{proof}
To prove monotonicity in $\rho$ first notice that
\begin{align*}
\rho(\Delta_1, i_1) > \rho(\Delta_2, i_2) \Rightarrow \mx{Z}^{(u)}(\Delta_1,i_1) \preceq  \mx{Z}^{(u)}(\Delta_2,i_2), \\
\rho(\Delta_1, i_1) > \rho(\Delta_2, i_2) \Rightarrow \bs{\Phi}^{(u)}(\Delta_1,i_1) \succeq  \bs{\Phi}^{(u)}(\Delta_2,i_2).
\end{align*}
and so
\begin{gather*}
\rho(\Delta_1, i_1) > \rho(\Delta_2, i_2)
\\
\Downarrow
\\
\begin{align*}
&\bs{\Phi}^{(u)}(\Delta_1,i_1) \left( \sum_{l = 1}^K \alpha_l^2 P_l \mx{Z}^{(u)}_l(\Delta_1, i_1) + \sigma_d^2 \mx{I}_{N_r} \right)^{-1}  \\
&~~\succeq \bs{\Phi}^{(u)}(\Delta_2,i_2) \left( \sum_{l = 1}^K \alpha_l^2 P_l \mx{Z}^{(u)}_l(\Delta_2, i_2) + \sigma_d^2 \mx{I}_{N_r} \right)^{-1}
\end{align*}
\\
\Downarrow
\\
\begin{align*}
&\textup{tr} \left(\bs{\Phi}^{(u)}(\Delta_1,i_1) \left( \sum_{l = 1}^K \alpha_l^2 P_l \mx{Z}^{(u)}_l(\Delta_1, i_1) + \sigma_d^2 \mx{I}_{N_r} \right)^{-1}\right)  \\
&\geq \textup{tr} \! \left( \!\bs{\Phi}^{(u)}(\Delta_2,i_2) \left( \sum_{l = 1}^K \alpha_l^2 P_l \mx{Z}^{(u)}_l(\Delta_2, i_2) + \sigma_d^2 \mx{I}_{N_r}\! \right)^{\!\!-1}\right)
\end{align*}
\\
\Downarrow
\\
\bar{\gamma}^{(u)}(\Delta_1,i_1) \geq \bar{\gamma}^{(u)}(\Delta_2,i_2).
\end{gather*}
Finally to prove convergence to 0 notice that
\begin{align*}
    \rho(\Delta, i) \rightarrow 0 &\Rightarrow \mx{Z}^{(u)}(\Delta_1,i_1) \rightarrow \mx{C}, \\
    \rho(\Delta, i) \rightarrow 0 &\Rightarrow \bs{\Phi}^{(u)}(\Delta_1,i_1) \rightarrow \mx{0}.
\end{align*}
And so, when $ \rho(\Delta, i) \rightarrow 0$ we have
\begin{align*}
    \bar{\gamma}^{(u)}&(\Delta,i) = \\
    &\textup{tr} \left(\bs{\Phi}^{(u)}(\Delta,i) \left( \sum_{l = 1}^K \alpha_l^2 P_l \mx{Z}^{(u)}_l(\Delta, i) + \sigma_d^2 \mx{I}_{N_r} \right)^{-1}\right)
\end{align*}
\begin{align*}
   & \stackrel{\rho(\Delta, i) \rightarrow 0}{\rightarrow}\textup{tr} \left(\mx{0} \left( \sum_{l = 1}^K \alpha_l^2 P_l \mx{C} + \sigma_d^2 \mx{I}_{N_r} \right)^{-1}\right) = 0.
\end{align*}
\end{proof}

\section{Proof of Proposition \ref{UpperB2}}
\label{Sec:P3}
\begin{proof}
From Theorem \ref{thm:upperbound} and \eqref{eq:SEku} the inequality follows.
For monotonicity, notice that $\rho_k(\Delta + 1, i) < \rho_k(\Delta, i )$ and
$\rho_k(\Delta + 1, i + 1) < \rho_k(\Delta, i )$.
Since by Proposition \ref{UpperB1} the upper bound of the \ac{SINR} is increasing with $\rho_k$ we have
\begin{align}
    \bar{\gamma}_k^{(u)}(\Delta + 1,i) &\leq \bar{\gamma}_k^{(u)}(\Delta,i) \nonumber \\
    \bar{\gamma}_k^{(u)}(\Delta + 1,i + 1) &\leq \bar{\gamma}_k^{(u)}(\Delta,i),
\end{align}
from which it follows that
\begin{align}
    \log\big(1 + \bar{\gamma}_k^{(u)}(\Delta + 1,i)\big) &\leq \log\big(1 + \bar{\gamma}_k^{(u)}(\Delta,i)\big)
    \label{eq:delta_step1}\\
    \log\big(1 + \bar{\gamma}_k^{(u)}(\Delta + 1,i + 1)\big) &\leq \log\big(1 + \bar{\gamma}^{(u)}(\Delta,i)\big).
    \label{eq:delta_step2}
\end{align}
Let $\ell = \arg\min_i \bar{\gamma}_k^{(u)}(\Delta + 1,i)$,
we then have
\begin{align}
 \nonumber &  \frac{1}{\Delta+1}  \times \sum_{i=1}^{\Delta+1} \log\big(1 + \bar{\gamma}_k^{(u)}(\Delta + 1,i)\big)
 \\ \nonumber
 & \leq  \frac{1}{\Delta} \times \left( \sum_{i=1}^{\ell-1} \log\big(1 + \bar{\gamma}_k^{(u)}(\Delta + 1,i)\big) \right. \\
    & ~~~~~~~~~~~~~~~~~~~~ + \left. \sum_{i=\ell+1}^{\Delta+1} \log\big(1 + \bar{\gamma}_k^{(u)}(\Delta + 1,i)\big) \right),
\end{align}
since on the right hand side we are removing the smallest term before calculating the mean.
Invoking \eqref{eq:delta_step1} and \eqref{eq:delta_step2} on the first and second sum, respectively, it follows that
\begin{align}
\nonumber
&   \frac{1}{\Delta + 1} \times \sum_{i=1}^{\Delta+1} \log\big(1 + \bar{\gamma}_k^{(u)}(\Delta + 1,i) \big)
\\ \nonumber
&  \leq  \frac{1}{\Delta + 1} \times \left( \sum_{i=1}^{\ell-1} \log\big(1 + \bar{\gamma}_k^{(u)}(\Delta,i)\big) \right. \\ \nonumber
&  ~~~~~~~~~~~~~~~~~~~~~+\left. \sum_{i=\ell}^{\Delta} \log\big(1 + \bar{\gamma}_k^{(u)}(\Delta,i)\big) \right)
\\
&=  \frac{1}{\Delta } \times \sum_{i=1}^{\Delta} \log\big(1 + \bar{\gamma}_k^{(u)}(\Delta,i)\big).
\end{align}
From which it follows that
\begin{align}
& \textup{SE}_k^{(u)}(\Delta+1) = \frac{\sum_{i=1}^{\Delta+1} \log\big(1 + \bar{\gamma}_k^{(u)}(\Delta + 1,i)\big)}{\Delta + 1}
\nonumber  \\
&~~~~~ \leq  \frac{\sum_{i=1}^{\Delta} \log\big(1 + \bar{\gamma}_k^{(u)}(\Delta ,i)\big)}{\Delta} = \textup{SE}_k^{(u)}(\Delta),
\label{eq:SEu}
\end{align}
that is $\textup{SE}_k^{(u)}(\Delta)$ is decreasing in $\Delta$.

To prove convergence to zero, recall from Proposition \ref{UpperB1} that
$\partial \bar{\gamma}_k^{(u)}(\Delta, i) / \partial \rho_k(\Delta, i) \geq 0$ and
\begin{align}
\nonumber
\rho_k(\Delta, i) \rightarrow 0 & \Rightarrow \bar{\gamma}_k^{(u)}(\Delta, i) \rightarrow 0
\end{align}
\begin{align}
&     \Rightarrow \log\big(1 + \bar{\gamma}_k^{(u)}(\Delta, i)\big) \rightarrow 0,
\label{eq:conv1}
\end{align}
where
\begin{align}
\nonumber
\rho_k(\Delta, i) = e^{2\bar{q}_k (\Delta + 1 + i)} + e^{2\bar{q}_k  i} + e^{2\bar{q}_k (\Delta + 1 - i)}.
\end{align}
We show that for any $\varepsilon > 0$, there is some $M$ such that
\begin{align}
 \textup{SE}^{(u)}(M) < \varepsilon.
\end{align}
Due to $\bar{q}_k < 0$, we have $\rho_k(\Delta, i) < \rho_k(1,1)$,
which implies
\begin{align}
    \log\big(1 + \bar{\gamma}_k^{(u)}(\Delta, i)\big) < \log\big(1 + \bar{\gamma}_k^{(u)}(1, 1)\big),
\end{align}
for all $\Delta$ and $i$.
Let $A \triangleq \log\big(1 + \bar{\gamma}_k^{(u)}(1, 1)\big)$
and
$N$ such that $N\varepsilon - 2A > 0$, and set
\begin{align}
\label{eq:epsilon}
    \epsilon \triangleq \frac{N\varepsilon - 2A}{N - 2}.
\end{align}

Since $\bar{q}_k < 0$, we have
\begin{align*}
\rho_k(\Delta, i) &< 3 \max(e^{2\bar{q}_k  (\Delta+1+i)},e^{2\bar{q}_k  i},e^{2\bar{q}_k  (\Delta+1-i)})
\\&= 3 e^{2\bar{q}_k  \min(\Delta+1+i,i,\Delta+1-i)},
\end{align*}
and it follows that for $\frac \Delta N \leq i \leq \frac{(N-1)\Delta}{N}$
\begin{align}
\rho_k(\Delta, i) < 3e^{2\bar{q}_k \frac \Delta N}.
\end{align}
Notice that by equation \eqref{eq:conv1} we can choose some large $M$,
such that
\begin{align}
&\frac M N \leq i \leq \frac{(N-1)M}{N}
    \Rightarrow
     \log(1 + \bar{\gamma}_k^{(u)}(M, i)) < \epsilon.
\end{align}
We can now show that when $M=\Delta$, then
$\textup{SE}_k^{(u)}(\Delta) < \varepsilon$.
To this end,
we split up the sum in the numerator of \eqref{eq:SEu}, that is
$\sum_{i=1}^{\Delta} \log(1 + \bar{\gamma}^{(u)}(\Delta ,i))$, into three terms,
and bound the first and third terms using the general upper bound $A$,
and the middle term by $\epsilon$:
\begin{align}
\nonumber
\textup{SE}_k^{(u)}(\Delta) &= \frac{\sum_{i=1}^{\Delta} \log(1 + \bar{\gamma}^{(u)}(\Delta ,i))}{\Delta}  \\ \nonumber
&=
\frac{\sum_{i=1}^{\Delta / N} \log(1 + \bar{\gamma}^{(u)}(\Delta ,i))}{\Delta} \\ \nonumber
&~~~~+
\frac{\sum_{i=\Delta / N + 1}^{(N - 1)\Delta / N} \log(1 + \bar{\gamma}^{(u)}(\Delta ,i))}{\Delta}
\\ \nonumber
&~~~~+
\frac{\sum_{i=(N-1)\Delta / N + 1}^{\Delta} \log(1 + \bar{\gamma}^{(u)}(\Delta ,i))}{\Delta}
\\ \nonumber
&<
\frac{(\Delta / N)A}{\Delta} + \frac{((N-2)\Delta / N)\epsilon}{\Delta} +
\frac{(\Delta / N)A}{\Delta}  \\ &=
\frac{2A + (N-2)\epsilon}{N} = \varepsilon,
\end{align}
where the last equation is due to the definition of $\epsilon$ in \eqref{eq:epsilon}, which completes the proof.
\end{proof}
\bibliography{sampling}
\end{document}